\documentclass[preprint,12pt]{elsarticle}

\usepackage{amsmath, amssymb, booktabs, appendix, hyperref, paralist, subcaption, caption}
\usepackage{array} 

\usepackage{setspace}
\usepackage{mathtools}
\usepackage{xcolor}
\usepackage{newtxtext}
\usepackage{setspace}
\usepackage{amsthm}
\allowdisplaybreaks[4]
\newtheorem{definition}{Definition}
\newtheorem{theorem}{Theorem}
\newtheorem{assumption}{Assumption}
\newtheorem{lemma}{Lemma}

\usepackage[ruled,linesnumbered]{algorithm2e}
\usepackage[top=2.54cm, bottom=2.54cm, left=1.8cm, right=1.8cm]{geometry}
\usepackage{multirow}

\journal{Journal of Multivariate Analysis}

\begin{document}

\begin{frontmatter}



\title{Joint Estimation of Edge Probabilities for Multi-layer Networks via Neighborhood Smoothing} 


\author[label1]{Yong He} 
\author[label1]{Zizhou Huang}
\author[label2]{Bingyi Jing}
\author[label3]{Diqing Li\corref{cor1}}

\address[label1]{Institute for Financial Studies, Shandong University}
 
\address[label2]{Southern University of Science and Technology}
	
\address[label3]{School of Statistics and Data Science, Zhejiang Gongshang University}
	
\cortext[cor1]{Corresponding author. \\ Email: \href{mailto:dqli@mail.zjgsu.edu.cn}{dqli@mail.zjgsu.edu.cn}}

\begin{abstract}
In this paper we focus on jointly estimating the edge probabilities for multi-layer networks. We define a novel multi-layer graphon, a ternary function in contrast to the bivariate graphon function in the literature by introducing an additional latent layer position parameter, which is model-free and covers a wide range of multi-layer networks. We develop a computationally efficient two-step neighborhood smoothing algorithm to estimate the edge probabilities of multi-layer networks, which requires little tuning and fully utilize  the similarity across  both network layers and  nodes. Numerical experiments demonstrate the advantages of our method over the existing state-of-the-art ones. A real Worldwide Food Import/Export Network dataset example is analyzed to illustrate the better performance of the proposed method over benchmark methods in terms of link prediction.
\end{abstract}



\begin{keyword}
Neighbourhood smoothing \sep Multi-layer networks \sep Graphon \sep Link prediction; 


\end{keyword}

\end{frontmatter}

\section{Introduction}
The rapid developments of science and technology have advanced the analysis of network structure data, which has broad application types such as  social networks, political networks, trading networks, biological networks, among others. Network data analysis is associated with a wide range of statistical inference problems, including link prediction, graph search and community detection, thereby enhancing the diversity of  network structures and expanding their applicability. Despite its significant progress in the last decades, most existing literature mainly focuses on the analysis of one individual large-scale network. However, with the explosion of data in real-world, network analysis is also evolving from the study of traditional single-layer network to more complex multi-layer networks, including multi-relational networks, dynamic networks and hyper-graph networks\citep{changepointdetection,Dynamicgraphon,YubaiYuan,hypergraphons,Zhao2019ChangepointDI}, which  capture richer interactions and more intricate system dynamics. 

The  focus of this paper is to jointly estimate the edge probabilities  in multi-layer networks, in which relationships between different nodes are reflected in multiple modalities. In contrast to single-layer networks, multi-layer networks can reveal more complex system dynamics and structural features. For instance, in neuroimaging studies, multi-layer networks can capture various interactions between regions of interest (ROIs) in the brain in different cognitive tasks.  Another motivating example is the global trade networks, where each network reveals different goods of trade among world-wide countries.  
Link prediction in multi-layer networks poses significant challenges, which necessitates not only taking into account the potential connections among nodes within a specific layer but also  the interactions  that occur across different layers. The existing work on link prediction mainly focus on single-layer. A naive method which directly utilize single-layer analysis methods to each layer would overlook crucial inter-layer information, thereby leading to inferior  link prediction  accuracy. This urges us developing innovative strategies that effectively integrate similarity information across multiple layers to enhance prediction accuracy.
The existing link prediction methods for multi-layer networks can be summarized into two categories. The first category predicts  links by calculating similarity score for non-adjacent nodes\citep{9634845}, which, however,   fails to provide a network generation mechanism and thereby having no  theoretical guarantee. The second category is on probabilistic network models, such as the multi-layer stochastic block models (MSBM). \citet{xu2023covariate} and \citet{jing2021community} propose to adopt Tucker decomposition to partition  multiple networks into communities, treating the networks as tensor objects. Although the edge probabilities can be estimated empirically combined with community detection algorithms, theoretical analysis would necessitate a strict structure assumption on the underlying probability matrix.

Graphon estimation has drawn growing attention for network analysis \citep{Wolfe2013NonparametricGE,Chan2014ACH,doi:10.1073/pnas.1400374111,minimaxGao2015,10.1214/16-AOS1497,10.1093/biomet/asab057}. By the Aldous-Hoover representation theorem \citep{aldous1981representations, hoover1979relations}, in the exchangeable network model, given the latent positions $\xi_i$ and $\xi_j$ of nodes $i$ and $j$ that can be  treated as $i.i.d.$ random variables  form Uniform $\left( 0,1 \right)$, then the link probability $P_{ij}$ can be expressed as a symmetric measurable function $f(\xi _i,\xi _j)$ that maps $\left[ 0,1 \right]^2$ to $\left[ 0,1 \right]$. For single-layer network, \citet{zhang2017estimating} proposed a neighborhood smoothing estimation method for piecewise-Lipschitz graphon functions, which achieves the best error rate among existing computationally feasible methods. However, it is still unclear how to adapt the neighborhood smoothing estimation method to deal with multi-layer network such that inter-layer information can also be utilized to improve estimation accuracy.

In this article we address the challenging problem building upon the work by \citet{zhang2017estimating}. We propose  a novel two-step neighborhood smoothing algorithm to estimate edge probabilities for multi-layer network which has little requirement on the network structures. We first adopt an adaptive neighborhood selection technique based on a well-designed similarity metric to identify neighborhood layers that are close to a specific layer. Then we locate the neighboring nodes to a specific node in each neighborhood layer. Our neighborhood smoothing estimator is finally obtained by averaging the adjacency matrices over both layers' and nodes' neighborhoods.

The main contributions of the work can be summarized as follows. Firstly  we introduce a novel  ternary graphon function $ f(\xi_i,\xi_j,\eta_k) $, where the third dimension $ \eta_k $ denotes the latent position of the $ k $th layer, and $ \xi_i $ denotes the latent position of node $i$ within that layer. The new ternary graphon  generalizes the traditional bivariate graphon function, and covers various multi-layer models as special cases. For instance, when $\eta_k$ remains constant, it degenerates to a traditional graphon. If $\eta_k$ is a time instant $t_k$, it can be viewed as a dynamic graphon \citep{Dynamicgraphon}. When  $\eta_k$ only takes two possible values, it can be used to detect change points in temporal networks\citep{changepointdetection}. When $f$ is piecewise-constant, it is in concordance with the multi-layer stochastic block model in \citet{Paul2015ConsistentCD}. In the case where $\eta_k$ is observed, the framework is equivalent to the multi-layer graphon model of \citet{Chandna2020NonparametricRF}.  Secondly, we propose a computationally efficient two-step neighborhood smoothing algorithm to estimate the edge probabilities of multi-layer networks, which requires little tuning and fully utilize  the similarity across  both network layers and  nodes. At last, we establish the convergence rate of our neighborhood smoothing estimator of probability matrices for multi-layer networks. The derived convergence rate  is faster than the convergence rate derived by \citet{zhang2017estimating} under certain condition on the scaling of $n$ and $K$.



\section{Methodology}
\subsection{Neighborhood smoothing estimator for multi-layer networks}

In this section, we study the problem of inferring  the generative mechanism of undirected multi-layer networks based on a single realization for each layer. 
Firstly we briefly review the problem setup for single-layer network. The data consist of the network adjacency matrix  $A\in\{0,1\}^{n\times n}$, where $n$ is the number of nodes and $A_{ij}=A_{ji}=1$ if there is an edge between nodes $i$ and $j$.
We assume each element $A_{ij}$  is independently generated from $\mathrm{Ber(}P_{ij})$ trials, where $P_{ij}$ is associated with the unknown latent variables $\xi _i$ and $\xi _j$ of nodes $i$ and $j$, as well as the underlying graphon function $f$, see more details in \cite{zhang2017estimating}. 

Throughout this paper, we focus on  multi-layer networks without cross-layer edges. The adjacency matrices $A^{k}\,(k=1,\ldots,K)$ of the multi-layer network are collectively denoted as a tensor $\boldsymbol{A}\in \{0,1\}^{n\times n\times K}$, where $n$ is the number of nodes and $K$ is the number of layers.  Each element  $A_{ijk}$ corresponds to the presence or absence of an edge between nodes $i$ and $j$ in the  $k$-layer.  The corresponding edge probabilities $\boldsymbol{P}=E\left( \boldsymbol{A} \right) \in \mathbb{R} ^{n\times n\times K}$ can be denoted as  $\boldsymbol{P}=\{ \left( P^k \right) _{n\times n}, \ k=1,\ldots,K \}$. Our goal is to estimate the edge probabilities $P_{ijk}$ based on the  realizations $A_{ijk}$. In the following we introduce a novel ternary graphon function to better illustrate the generative mechanism of $ A_{ijk} $, which generalizes the traditional bivariate graphon function with an additional  latent layer position parameter. In detail, we assume that the observed $ A_{ijk} $ are  independent  $ \mathrm{Ber}(P_{ijk}) $ trials, where $ P_{ijk} = f(\xi_i, \xi_j, \eta_k) $. The edge probabilities $ f(\xi_i, \xi_j, \eta_k) $ is a ternary graphon function that characterizes the multi-layer network, which is formally defined as follows.

\begin{definition}[Multi-layer Graphon]\label{definition1}
	The multi-layer network graphon function $f(\xi_i, \xi_j, \eta_k):[0,1]^2\times [0,1]\rightarrow [0,1]$ is a ternary function such that $P_{ijk}=f\left( \xi _i,\xi _j,\eta _k \right)$, where $\xi_i$, $\eta_k$ are independently sampled from distribution $P_{\xi}$ and $P_{\eta}$ respectively, and both $P_{\xi}$ and $P_{\eta}$ are supported on $\left[ 0,1 \right]$.
\end{definition}

For a fixed layer $k$ with a given latent parameter $\eta_k$, the connection probabilities are given by $P_{ijk} =f_k(\xi_i, \xi_j)=f(\xi_i, \xi_j, \eta_k)$. Since $\eta_k$ is fixed for the entire layer and the node-specific latent variables $\xi_i$ are i.i.d., the function $f(\cdot, \cdot, \eta_k)$ acts as a standard, symmetric graphon for that specific layer. Therefore, conditional on the layer parameter $\eta_k$, the network in layer $k$ is an exchangeable graph, and its existence and representation are explicitly guaranteed by the Aldous-Hoover theorem \citep{aldous1981representations}. 

While Defination \ref{definition1} of multi-layer graphon is similar to the defination in \cite{Chandna2020NonparametricRF}, there is a fundamental distinction in how layer information is modeled. The approach in \cite{Chandna2020NonparametricRF} relies on observed covariates $z_l$ as proxies for layer characteristics and employs a flattened graphon that averages over the network-specific dimension. In contrast, our model treats the layer variable $\eta_k$ as latent, making it suitable for inferring hidden layer structures without requiring observed covariates. In the following, we adopt the common uniform distribution on  $\left[ 0,1 \right]$ for  $P_{\xi}$ and $P_{\eta}$ in Definition \ref{definition1}. 

Next, we propose a new two-step neighborhood smoothing algorithm, incorporating interactions information across multiple layers to achieve more accurate estimation. 
Intuitively, when two networks exhibit similar connection patterns, their corresponding probability matrices will be similar. Specifically, for $P_{ijk}=f(\xi _i,\xi _j,\eta _k)$, when $\eta _k$ is close to another layer position variable $\eta_{k^\prime}$, we would expect that $P_{ijk}$ and $P_{ijk^\prime}$ still have similar values with some smoothness in $f$. Therefore, we define the neighborhood set of the $k$th layer network as $\mathcal{N}^{k}=\{k^{\prime}:P_{\cdot,\cdot,k^{\prime}}\approx P_{\cdot,\cdot,k}\}$.
Parallelly, for similar link patterns between nodes within a specific layer (i.e., fixing $\eta_k$), we observe that for $P_{ijk}=f(\xi _i,\xi _j,\eta _k)$, if $\xi _i$ and $\xi_{i^\prime}$ are relatively close, then $f\left( \xi _i,\cdot ,\eta _k \right)$ should be close to $f\left( \xi_{i^\prime},\cdot ,\eta _k \right)$, leading to  $P_{i,\cdot,k}\approx P_{i^\prime,\cdot,k}$. This indicates that the $i$th row and the $i^\prime$th row in the probability matrix of the $k$th layer network are fairly comparable. As a result, the adjacency matrix likewise has a corresponding similarity, so that $A_{i,\cdot,k}\approx A_{i^\prime,\cdot,k}$, suggesting that the configurations in these two rows of the adjacency matrix are comparable.

The set of neighboring nodes of node $ i $ within a neighboring network in the $k$th layer is denoted as $\mathcal{N}_{i}^{k^{\prime}}$, where $ k^{\prime} \in \mathcal{N}^k $. Thus, we can express $\mathcal{N} _{i}^{k^{\prime}}=\{i^{\prime}:P_{i^{\prime}\cdot k^{\prime}}\approx P_{i\cdot k^{\prime}}\}$.
The method for selecting $\mathcal{N}_{i}^{k^{\prime}}$ and $\mathcal{N}^{k}$ will be detailed in the next subsection. After identifying all neighboring network layers $\mathcal{N}^{k}$ for  the $k$th layer, and locating the neighborhood  $\mathcal{N} _{i}^{k^{\prime}}$ for node $i$ in the layer $k^{\prime}\in \mathcal{N} ^k$, we then estimate $P_{ijk}$ by a neighborhood smoothing procedure such that
\begin{equation}\label{equation::MNSestimator}
	\tilde{P}_{ijk}=\frac{\sum_{k^{\prime}\in \mathcal{N} ^k}{\sum_{i^{\prime}\in \mathcal{N} _{i}^{k^{\prime}}}{A_{i^{\prime}jk^{\prime}}}}}{\sum_{k^{\prime}\in \mathcal{N} ^k}{|}\mathcal{N} _{i}^{k^{\prime}}|}.
\end{equation}
The denominator $\sum_{k' \in \mathcal{N}^k} |\mathcal{N}_{i}^{k'}|$ is the total number of  neighboring nodes that are similar to node $i$ discovered across all neighboring networks in the $k$th layer.
When the network is symmetric, we can utilize the following symmetric estimator:
\begin{equation*}
	\hat{P}^k=\left[ \tilde{P}^k+\left( \tilde{P}^k \right) ^{\mathrm{T}} \right] /2.
\end{equation*}

An alternative approach is to smooth the elements of the adjacency matrix over $\sum_{k^{\prime}\in \mathcal{N} ^k}^{}{\mathcal{N} _{i}^{k^{\prime}}}\times \mathcal{N} _{j}^{k^{\prime}}$. Nevertheless, in the subsequent algorithms, we can vectorize and matrixize the expression \eqref{equation::MNSestimator}, significantly alleviating computational burden. We keep the potential relationships between layers and avoid fragmentation of the multi-layer networks by applying the neighborhood smoothing method on both the layers and nodes collectively. Additionally, we will demonstrate in the upcoming part that using quantile selection for neighbors offers more flexibility and adaptability.

\subsection{Adaptive neighborhood selection for layers and nodes}

In this section we further discuss how to select neighboring layers $\mathcal{N}^k$ and nodes $\mathcal{N}_{i}^{k'}$ in the multiple network. Notice that when $f\left( \cdot, \cdot, \eta_k \right) \approx f\left( \cdot, \cdot, \eta_\ell \right)$, the $k$th layer of the network exhibits similar probabilistic behavior to the $\ell$th layer. This implies that $ P^{k} \approx P^{\ell} $,  further indicating that the adjacency matrices $ A^{k}$ and $A^{\ell}$ are also similar. We define layers with this property as neighboring layers. 

We use the $\ell_2$ distance between the graphon function slices of the layers to identify the neighboring layers of the $k$th layer, defined as follows:
\begin{equation*}
		d\left( k,k^{\prime} \right) =\left\| f(\cdot ,\cdot ,\eta _k)-f(\cdot ,\cdot ,\eta _{k^{\prime}}) \right\| _2
		=\left\{ \int_0^1{\int_0^1{\left| f(u,v,\eta _k)-f(u,v,\eta _{k^{\prime}}) \right|^2}}\mathrm{d}u\mathrm{d}v \right\} ^{1/2}.
\end{equation*}
Other norm distances can also be utilized for this measurement. However, the $\ell_2$ norm is particularly convenient for calculation and theoretical analysis.

Fixing $\eta$, we treat $f(u,v,\eta)$ as a traditional bivariate graphon function such that $f_k\left( u,v \right) =f(u,v,\eta _k)$ and $f_{k^{\prime}}\left( u,v \right) =f(u,v,\eta _{k^{\prime}})$. For integrable bivariate functions $f_k$ and $f_{k^{\prime}}$ defined on $[0,1]^2$, we can define the inner product
\begin{equation*}
	\langle f_k,f_{k^{\prime}}\rangle =\int_0^1{\int_0^1{f_k}}(u,v)f_{k^{\prime}}(u,v)\mathrm{d}u\mathrm{d}v.
\end{equation*}
Thus, we can rewrite $d^2(k, k^{\prime})$ as
\begin{equation}\label{equation::original_distance}
		d^2(k,k^{\prime})=\langle f_k-f_{k^{\prime}},f_k-f_{k^{\prime}}\rangle	=\langle f_k,f_k\rangle +\langle f_{k^{\prime}},f_{k^{\prime}}\rangle -2\langle f_k,f_{k^{\prime}}\rangle . 
\end{equation}
The third term can be estimated by $\langle A^k,A^{k^{\prime}}\rangle /n^2$ since different layers  are independent. However, The first two terms cannot be directly estimated via the same approach. To explain this, we assume $A^k=f_k+\epsilon_k$, where $\epsilon_k$ is error dependent of $f_k$ and $E(\epsilon_k)=0$. The cross-term $\langle f_k,f_{k^{\prime}}\rangle$ can be unbiasedly estimated by $\langle A^k,A^{k^{\prime}}\rangle /n^2$ because the noise $\epsilon_k$ and $\epsilon_{k'}$ are independent. However, estimating the self-terms  $\langle f_k,f_k\rangle$ using $\langle A^k, A^k \rangle / n^2$ introduces a significant bias. Since $A^k = f_k + \epsilon_k$, the expectation contains a variance term $E\langle \epsilon_k, \epsilon_k \rangle > 0$, which creates a large bias that is difficult to remove. To address this issue, assuming that for $f(\cdot ,\cdot ,\eta _k)$ with $k=1,\ldots,K$, we can find $\tilde{k}\ne k$ such that $\left| \eta _k-\eta _{\tilde{k}} \right|\le e_K$ with high probability where the sequence $e_K$ is a error rate function of $K$. We further assume that the function $f$ is Lipschitz continuous with a Lipschitz constant of 1, then $\left\| f\left( \cdot ,\cdot ,\eta _k \right) -f(\cdot ,\cdot ,\eta _{\tilde{k}}) \right\| _2\le e_K$, so we can approximate  $\left. \langle f_k,f_k \right. \rangle $ by $\left. \langle f_k,f_{\tilde{k}} \right. \rangle $, where the latter can now be estimated by $\left. \langle A^k,A^{\tilde{k}}\rangle /n^2 \right. $. The basic idea is to replace the intractable self-inner product of a graphon slice with the tractable cross-inner product of two different but nearly identical slices. The same technique can be used to approximate the second term in \eqref{equation::original_distance}.
We can rearrange $d^2(k,k^{\prime})$ as follows:
\begin{equation*}
		d^2\left( k,k^{\prime} \right) =\left. \langle f_k-f_{k^{\prime}},f_k \right. \rangle -\left. \langle f_k-f_{k^{\prime}},f_{k^{\prime}} \right. \rangle
		\le 2\max_{\ell \ne k,k^{\prime}} \left| \left. \langle f_k-f_{k^{\prime}},f_{\ell} \right. \rangle \right|+2e_K.
\end{equation*}
It is not necessary to provide the exact estimates for each pair of $d(k,k^{\prime})$ in the neighborhood selection process. Alternatively, we can find the neighbor set by comparing an approximate upper bound of  $d(k,k^{\prime})$, which are more tractable.  The inner product on the right-hand side of the above equation can be obtained by :
\begin{equation}\label{equation::layer_distance}
	\tilde{d}^2\left( k,k^{\prime} \right) =\max_{\ell \ne k,k^{\prime}} \left| \left. \langle A^k-A^{k^{\prime}},A^{\ell} \right. \rangle \right|/n^2=\max_{\ell \ne k,k^{\prime}} \left| \mathrm{tr}\left[ \left( A^k-A^{k^{\prime}} \right) ^{\mathrm{T}}A^{\ell} \right] \right|/n^2.
\end{equation}
Therefore, we can estimate the layer distance approximately without knowing graphon $f$ and latent variable $\eta$.
Intuitively, the neighborhood $\mathcal{N} ^k $ should include $k^{\prime}$ with small $\tilde{d}(k, k^{\prime})$. To formalize this, let $q_k(h_1)$ denote the $h_1$th sample quantile of the set $\{ \tilde{d}(k, k^{\prime}) : k^{\prime} \neq k \}$, where $h_1$ is a tuning parameter, and set
\begin{equation}\label{equation::layer_neighbor}
	\mathcal{N} ^k=\{k^{\prime}\ne k:\tilde{d}(k,k^{\prime})\le q_k(h_1)\}.
\end{equation}

The choice of $h_1$ will be guided by both the theory in § \ref{2.4}, which suggests the order of $h_1$, and empirical performance, which suggests the constant factor. More details can be found in the supplementary material.

For each neighboring layer $ k^{\prime} \in \mathcal{N}^k$, we seek to identify the set of neighboring node $\mathcal{N} _{i}^{k^{\prime}}$ for node $i$. This involves finding node $i^{\prime}$ in layer $k^{\prime}$ that exhibit similar probabilistic behavior to node $i$ within the same layer. Analogous to identifying node neighbors in a  single-layer network, we adopt a similar selection criterion:
\begin{equation}\label{equation::MNSnode_neighbor}
	\mathcal{N} _{i}^{k^{\prime}}=\{i^{\prime}:P_{i^{\prime}\cdot k^{\prime}}\approx P_{i\cdot k^{\prime}}\}=\left\{ i^{\prime}\ne i:\tilde{d}\left( i,i^{\prime} \right) \le q_i(h_2) \right\}
\end{equation}
as that in \citet{zhang2017estimating}. Unlike \citet{zhang2017estimating}, whose quantile depends solely on $n$, the quantiles $h_1$ and $h_2$ used in our method for selecting neighbor layers and neighbor nodes are functions of both $n$ and $K$. Intuitively, we narrow the range of neighbor nodes while placing greater emphasis on information from neighbor layers. As a result, our approach achieves a faster convergence rate compared to that of \citet{zhang2017estimating}.

\subsection{Multi-layer neighborhood smoothing algorithm}

In this section, we provide a detailed exposition of the algorithm for Multi-layer Neighborhood Smoothing (MNS). 

To calculate the average of the adjacency matrix elements associated with neighbor nodes in each neighboring layer, we begin by evaluating the term $\sum_{i^{\prime} \in \mathcal{N}_{i}^{k^{\prime}}} A_{i^{\prime}jk^{\prime}}$ and $\mathcal{N}_{i}^{k^{\prime}}$ as outlined in \eqref{equation::MNSestimator}. This involves iterating through all relevant adjacency matrices and calculating $\sum_{i^{\prime} \in \mathcal{N}_{i}} A_{i^{\prime}j}$ for each node in each layer. Subsequently, we identify the neighboring layers for each given layer and compute the average of the corresponding adjacency matrix elements.

The core of the algorithm focus on measuring the similarity between slices of the adjacency tensor, which captures the distance between neighboring nodes and layers. To improve computational efficiency, the distance measurements  are vectorized for matrix operations. Specifically,  the node distance in layer $k$ can be formulated  in terms of Chebyshev distance as follows:
\begin{equation*}\label{equation::node_chebyshev}
	\tilde{d}^2_k\left( i,i^{\prime} \right) =\max_{s\ne i,i^{\prime}} \left| \left. \langle A_{i\cdot}^{k}-A_{i^{\prime}\cdot}^{k},A_{s\cdot}^{k} \right. \rangle \right|/n=\max_{s\ne i,i^{\prime}} \left| \left[ \left( A^k \right) ^2/n \right] _{is}-\left[ \left( A^k \right) ^2/n \right] _{i^{\prime}s} \right|.
\end{equation*}
From this process, we can derive a symmetric distance matrix $D^k\in \mathbb{R} ^{n\times n}$ for $k=1,\ldots,K$, where $D_{ij}^{k}$  represents the Chebyshev distance $\tilde{d}^2_k\left( i, j \right)$ between the $i$th and $j$ th rows of  $\left( A^k \right) ^2/n$. By comparing each element of $D^{k}$ with the specified quantile $q_i\left( h_2 \right)$, we can create the tensor $\boldsymbol{B}=\left\{ B^k,k=1,...,K \right\} \in \left\{ 0,1 \right\} ^{n\times n\times K}$ denoted by the node neighborhood matrices $B^k\left( k=1,...,K \right)$ collectively, where $B_{i j}^{k}=1$ indicates that node $j$ is a neighbor of node $i$ in layer $k$. The number of non-zero elements in the $i$th column of $B^{k}$ indicates the number of neighboring nodes for node $i$. The product $A^{k} B^{k}$ conveniently records the $\sum_{i^{\prime}\in \mathcal{N} _{i^{\prime}}^{k}}{A_{i^{\prime}jk}}$ for each node in the $k$th layer, facilitating the computation of aggregated adjacency information within the neighborhood. These products could be collectively denoted as tensor $\boldsymbol{C}=\left\{ A^kB^k,k=1,...,K \right\} \in \mathbb{R} ^{n\times n\times K}$.

Similarly, the layer distance defined in \eqref{equation::layer_distance} can be rewritten as:
\begin{equation*}\label{equation::layer_chebyshev}
	\begin{split}
		\tilde{d}^2\left( k,k^{\prime} \right) &=\max_{\ell \ne k,k^{\prime}} \left| \mathrm{tr}\left[ \left( A^k-A^{k^{\prime}} \right) ^{\mathrm{T}}A^{\ell} \right] \right|/n^2\\
		&=\max_{\ell \ne k,k^{\prime}} \left| \mathrm{tr}\left[ \left( A^k \right) ^{\mathrm{T}}A^{\ell} \right] -\mathrm{tr}\left[ \left( A^{k^{\prime}} \right) ^{\mathrm{T}}A^{\ell} \right] \right|/n^2.
	\end{split}	
\end{equation*}

This yields the layer distance matrix $D^{\mathrm{layer}}\in \mathbb{R}^{K\times K}$, where $D_{kk^{\prime}}^{\mathrm{layer}}$ captures the Chebyshev distance $\tilde{d}^2\left( k,k^{\prime} \right)$ between the $k$th row and the $k^{\prime}$th row of $T$, defined as:
$$
T=\left[ \begin{matrix}
	\mathrm{tr}\left[ \left( A^1 \right) ^{\mathrm{T}}A^1 \right]&		\cdots&		\mathrm{tr}\left[ \left( A^1 \right) ^{\mathrm{T}}A^K \right]\\
	\vdots&		\ddots&		\vdots\\
	\mathrm{tr}\left[ \left( A^K \right) ^{\mathrm{T}}A^1 \right]&		\cdots&		\mathrm{tr}\left[ \left( A^K \right) ^{\mathrm{T}}A^K \right]\\
\end{matrix} \right] /n^2.
$$
With the threshold  $q_k\left( h_1 \right)$ in \eqref{equation::layer_neighbor}, we can derive the layer neighborhood matrix $B^{\mathrm{layer}} \in \{0, 1\}^{K \times K}$, where $B_{k\ell}^{\mathrm{layer}} = 1$ indicates layer $\ell$ is a neighbor of $k$ such that $D_{k\ell}^{\mathrm{layer}}\le q_k\left( h \right)$.
In addition, we can use $\boldsymbol{C}\times _3B_{\cdot k}^{\mathrm{layer}}$ to record $\sum_{k^{\prime}\in \mathcal{N} ^k}{A^{k^{\prime}}}B^{k^{\prime}}=\sum_{k^{\prime}\in \mathcal{N} ^k}{\sum_{i^{\prime}\in \mathcal{N} _{i^{\prime}}^{k}}{A_{i^{\prime}jk^{\prime}}}}$ for each node in the $k$th layer, where $B_{\cdot k}^{\mathrm{layer}}$ is the $k$th column of $B^{\mathrm{layer}}$, and $\times _3$ denotes the mode-3 tensor product. We also use $1^{\mathrm{T}} \cdot \left[ \boldsymbol{B} \times_3 B_{\cdot k}^{\text{layer}} \right]$ to effectively sums the elements along the corresponding columns of $\boldsymbol{B} \times_3 B_{\cdot k}^{\text{layer}}$, recording  $\sum_{k^{\prime}\in \mathcal{N} ^k}{|\mathcal{N} _{i}^{k^{\prime}}|}$ for each node, where $1^{\mathrm{T}}$ is a row vector of all ones  and $\boldsymbol{B}\times _3B_{\cdot k}^{\mathrm{layer}}$  would represent $\sum_{k^{\prime}\in \mathcal{N} ^k}{B^{k^{\prime}}}$. The detailed algorithm for MNS is summarized in the Algorithm \ref{algo_MNS} below.
\begin{algorithm}
	\SetKwData{Left}{left}\SetKwData{This}{this}\SetKwData{Up}{up}
	\SetKwFunction{Union}{Union}\SetKwFunction{FindCompress}{FindCompress}
	\KwIn{Adjacency matrix $A^{k}\left( k=1,\ldots ,K \right) $.}
	\KwOut{Probability matrix $P^{k}\left( k=1,\ldots ,K \right) $.}
	\BlankLine
	\For{$ k=1,\ldots ,K $ }{
		\textnormal{Obtain $D^{k}$ of each layer by calculating the Chebyshev distance of $\left( A^k \right) ^2/n$.}\\
		\textnormal{Obtain $B^{k}$ controlled by threshold  $q_k\left( h_2 \right)$.}\\
		\textnormal{Calculate $A^kB^k$ to record $\sum_{i^{\prime}\in \mathcal{N} _{i^{\prime}}^{k}}{A_{i^{\prime}jk}}$ for each node in each layer.}\\
	}
	\textnormal{Obtain $D^{\mathrm{layer}}$ by calculating the Chebyshev distance between different rows of $T$.}\\
	\textnormal{Obtain $B^{\mathrm{layer}}$ controlled by threshold  $q_k\left( h_1 \right)$.}\\
	\textnormal{Create $\boldsymbol{B} = \left\{ B^k,k=1,...,K \right\}$ and calculate $\boldsymbol{C}=\left\{ A^kB^k,k=1,...,K \right\} $.}\\
	\For{$ k=1,\ldots ,K $}{
		\textnormal{Calculate $\boldsymbol{C}\times _3B_{\cdot k}^{\mathrm{layer}}$ to record $\sum_{k^{\prime}\in \mathcal{N} ^k}{A^{k^{\prime}}}B^{k^{\prime}}$ for each node. }\\
		\textnormal{Calculate $1^{\mathrm{T}}\cdot \left[ \boldsymbol{B}\times _3B_{\cdot k}^{\mathrm{layer}} \right]$ to record $\sum\nolimits_{k^{\prime}\in \mathcal{N} ^k}^{}{|\mathcal{N} _{i}^{k^{\prime}}|}$.}\\
		\textnormal{Obtain $\tilde{P}^k$ by applying $L_1$ regularization to $\boldsymbol{B}\times _3B_{\cdot k}^{\mathrm{\mathrm{layer}}}$ using $1^{\mathrm{T}}\cdot \left[ \boldsymbol{B}\times _3B_{\cdot k}^{\mathrm{\mathrm{layer}}} \right] $.}\\
		\textnormal{Obtain  $\hat{P}^k=\left[ \tilde{P}^k+\left( \tilde{P}^k \right) ^{\mathrm{T}} \right] /2$.}\\
	}
	\caption{Multi-layer Neighborhood Smoothing algorithm}\label{algo_MNS}
\end{algorithm}

\section{Theoretical properties}\label{2.4}

In this section, we establish the consistency  of the MNS estimator. To this end, we first introduce the piecewise Lipschitz condition for multi-layer graphon family.

\begin{definition}[Piecewise Lipschitz Multi-graphon Family]
	For any positive $\delta>0$, $L_n>0$, $L_K>0,$ denote $L=\left( L_n,L_K \right) ^{\mathrm{T}}$, let $\mathcal{F} _{\delta ;L}$ be a family of piecewise Lipschitz graphon functions $f:\left[ 0,1 \right] ^3\rightarrow \left[ 0,1 \right]$, which has the following properties:
	
	\vspace{0.1em}
	
	\begin{inparaenum}[(i)]
		\item  There exists an integer $R\ge 1$ and a sequence $0=x_0<\cdots <x_R=1$ satisfying $\min_{0\le r\le R-1} \left( x_{r+1}-x_r \right)\\ >\delta $. Additionally, there exists an integer $T\ge 1$ and a sequence $0=z_0<\cdots <z_T=1$ satisfying $\min_{0\le t\le T-1} \left( z_{t+1}-z_t \right) >\delta $; 
		
		\vspace{0.1em}
		
		\item There exists some constant $L_n > 0$ such that both $\left| f\left( u_1,v,w \right) -f\left( u_2,v,w \right) \right|\le L_n\left| u_1-u_2 \right|$ and $\left| f\left( u,v_1,w \right) -f\left( u,v_2,w \right) \right|\le L_n\left| v_1-v_2 \right|$ hold for all $u,u_1,u_2\in \left[ x_r,x_{r+1} \right) $ and $v,v_1,v_2\in \left[ x_s,x_{s+1} \right)$ where $0\le r,s\le R-1$; 
		
		\vspace{0.1em}
		
		\item  There exists some constant $L_K > 0$ such that $\left| f\left( u,v,w_1 \right) -f\left( u,v,w_2 \right) \right|\le L_K\left| w_1-w_2 \right|$ holds for $w,w_1,w_2\in \left[ z_t,z_{t+1} \right)$ where $0\le t\le T-1$.
	\end{inparaenum}
\end{definition}

Then for any $P^{ k }\in \mathbb{R} ^{n\times n}\left( k=1,\ldots,K \right) $, we have the following normalized Frobenius norm error rate bound. 

\begin{theorem}\label{theorem1}
	Assume $L_n$ and $L_K$ in $L$ are global constants, and $\delta =\delta (n,K)$ depends on $n$, $K$ satisfying ~\\ $\lim_{n,K\rightarrow \infty} \delta /\left\{ \left( nK \right) ^{-1}\log n \right\} ^{1/3}\rightarrow \infty $ and $K\lesssim  n^2\log n$, then the estimator $\tilde{P}$ defined in \eqref{equation::MNSestimator}, with neighborhood $\mathcal{N}^{k}$, $\mathcal{N} _{i}^{k^{\prime}}$defined in \eqref{equation::layer_neighbor} and \eqref{equation::MNSnode_neighbor}, where $h_1= \max\{C \left\{ \left( nK \right) ^{-1}\log n \right\} ^{1/3}, 1/K\}$ and $h_2=\max\{\left\{C \left( nK \right) ^{-1}\log n \right\} ^{1/3}, C^{3/2}(\log n/n)^{1/2}\}$ for any global constant $C > 0$ satisfies
	\begin{equation*}
		\max_{f\in \mathcal{F} _{\delta ;L}} \mathrm{pr}\left\{ \frac{1}{n^2}\left\| \tilde{P}^k-P^k \right\| _{F}^{2}\ge C_1\left( \frac{\log n}{nK} \right) ^{1/3}+C_2\left(\frac{\log n}{n}\right)^{1/2} \right\} \le n^{-\gamma}
	\end{equation*}
	for $\gamma>0$, where $C_1$ and $C_2$ are positive global constants. 
\end{theorem}

Theorem \ref{theorem1} establishes the convergence rate for the MNS estimators of the probability matrices in multi-layer networks. The convergence rate benefits from larger  $K$ (the number of layers) and  $n$ (the number of nodes). When $K\gtrsim \left( n/\log n \right) ^{1/2}$, the incorporation of auxiliary information from similar layers within the multi-layer network would enhance the estimation accuracy. In this case, the convergence rate is $O\left( \left\{ \log n/\left( nK \right) \right\} ^{1/3} \right)$, which is faster than the convergence rate $O\left( \left( \log n/n \right) ^{1/2} \right)$ derived by \citet{zhang2017estimating}. For the case when $K\lesssim \left( n/\log n \right) ^{1/2}$, the number of layer variables $\eta_k$'s become too sparse to effectively utilize information from neighboring layers and the information contained within a single-layer network is enough to guarantee a reliable and  effective estimation. In this case, we set $h_1=1/K$ and $h_2= C^{3/2}(\log n/n)^{1/2}$, meaning that only the information from the same layer is used to estimate $P^k$. Under this configuration, our method reduces to that of \citet{zhang2017estimating}. It is important to note that Theorem \ref{theorem1} requires the number of layers $K$ not to be excessively large such that $K\lesssim  n^2\log n$. This condition is met in the vast majority of practical scenarios. For example, in a network with $n=100$, the theorem holds as long as $K\le 46051$—a notably mild requirement. On the other hand, for the same network  with $n=100$, the proposed method can achieve a faster convergence rate than \citet{zhang2017estimating} once $K>4$. Therefore, the proposed method exhibits strong applicability in real-world settings. For the case where $K\gtrsim n^2\log n$, the following theorem shows that the convergence rate of the estimator no longer benefits from further increases in $K$.

\begin{theorem}\label{theorem2}
		Assume $L_n$ and $L_K$ in $L$ are global constants, and $\delta =\delta (n,K)$ depends on $n$, $K$ satisfying ~\\ $\lim_{n,K\rightarrow \infty} \delta /\left\{ \left( nK \right) ^{-1}\log n \right\} ^{1/3}\rightarrow \infty $ and $K>  n^2\log n$, then the estimator $\tilde{P}$ defined in \eqref{equation::MNSestimator}, with neighborhood $\mathcal{N}^{k}$, $\mathcal{N} _{i}^{k^{\prime}}$defined in \eqref{equation::layer_neighbor} and \eqref{equation::MNSnode_neighbor} and $h_1=h_2=C\left\{ \left( nK \right) ^{-1}\log n \right\} ^{1/3}$ for any global constant $C > 0$ satisfies
	\begin{equation*}
		\max_{f\in \mathcal{F} _{\delta ;L}} \mathrm{pr}\left\{ \frac{1}{n^2}\left\| \tilde{P}^k-P^k \right\| _{F}^{2}\ge \frac{\tilde{C}}{n^{1-\alpha}} \right\} \le n^{-\gamma}
	\end{equation*}
	for some $\gamma>0$, $\tilde{C}>0$ and any positive $\alpha<1$. 
\end{theorem}

Theorem \ref{theorem2} shows that when the number of layers becomes significantly large, such that $K\gtrsim n^2\log n$, the convergence rate reaches $O\left( n^{\alpha-1} \right)$, and the MNS estimator will no longer benefit from the increase of the number of layers. This occurs because when $K$ is sufficiently large, the layer-wise latent variables become adequately dense.  With such a massive number of networks, the information required to estimate $P^k$ via the smoothing across the layer becomes essentially perfect. The error term associated with $K$ becomes negligible. The rate $\frac{1}{n^{1-\alpha}}$ effectively corresponds to the estimation rate for a single graphon given $n$ nodes, which hit the "floor" of optimal rate for a single graphon $O\left( \log n/n \right)$ discussed in \citet{minimaxGao2015}. the limitation is no longer the lack of network layers $K$, but the granularity $n$ of the network itself. In the simulation studies, we also demonstrate that when $K$ becomes sufficiently large, the estimation error does not further decrease with increasing $K$.

\section{Numerical experiments}
\subsection{Simulations}

In this section, we conduct simulation studies to examine the empirical performance of the MNS estimators for multi-layer networks  in finite samples. To justify the broad applicability of our method, we generate the multi-layer networks from the four types of graphons listed in Table \ref{table:graphon} with different features, including low-rank, degree monotonicity and local structure. Their characteristics will be listed as follows. 

\begin{table}[htbp]
	\centering
	\caption{Synthetic multi-layer graphons}{	
		\begin{tabular}{cc}
			\toprule
			Graphon & Function $f(u,v,w)$\\
			 \hline
			1 & $m/\{(M+1)(w+1) \}$ if $u, v\in((m-1)/M, m/M)$,\\
			& $0.3/(M+1)$ otherwise; $M=\lfloor \log n\rfloor$\\
			2 & $\sin \{5\pi (u+v-w)+1 \}/2+0.5$\\
			3 & $1-[1+\exp \{15(0.8|u-v|)^{1/(1+w)}-0.1 \} ]^{-1}$\\
			4 & $(u^2+v^2)/3\cos \{w/(u^2+v^2) \}+0.15$\\
			\bottomrule
	\end{tabular}}\label{table:graphon}
\end{table}

\begin{itemize}
	\item  \textbf{Graphon 1}: This structure consists of $M=\lfloor \log n \rfloor$ blocks in each layer, each with distinct edge probabilities, while the probabilities across the blocks are generally low.
	
	\item  \textbf{Graphon 2}: This structure lacks monotonicity in node degrees, indicating that the degree distribution does not follow a consistent increasing or decreasing pattern.
	
	\item  \textbf{Graphon 3}: This is a diagonal dominant matrix, where the diagonal elements are significantly larger than the off-diagonal elements.
	
	\item  \textbf{Graphon 4}: This is a full-rank matrix characterized by structures at different scales, which is challenging for the  low-rank approximation methods.
\end{itemize}

\begin{figure}[h!]
	\centering
	\includegraphics[width=0.47\textwidth]{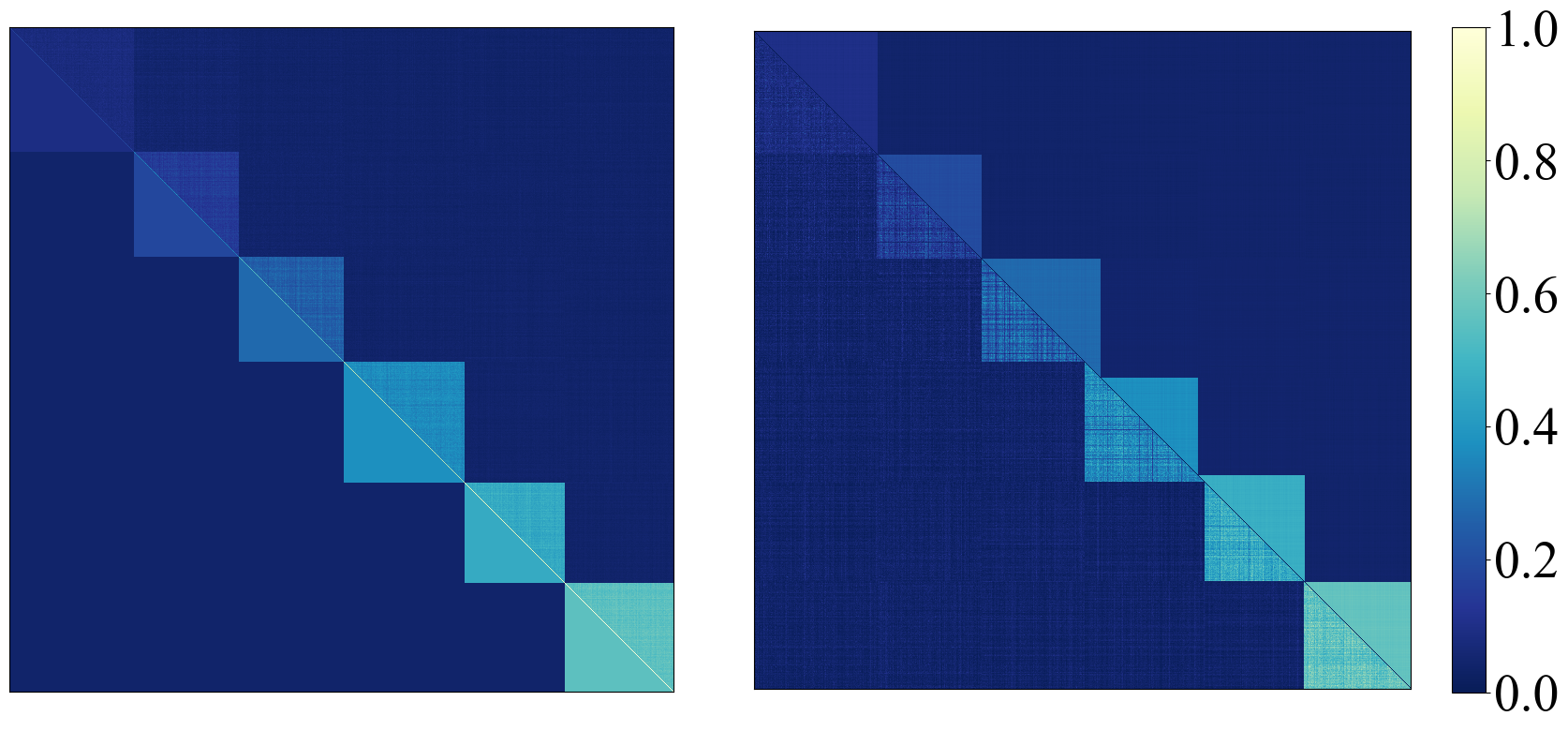} 
	\hspace{1em} 
	\includegraphics[width=0.47\textwidth]{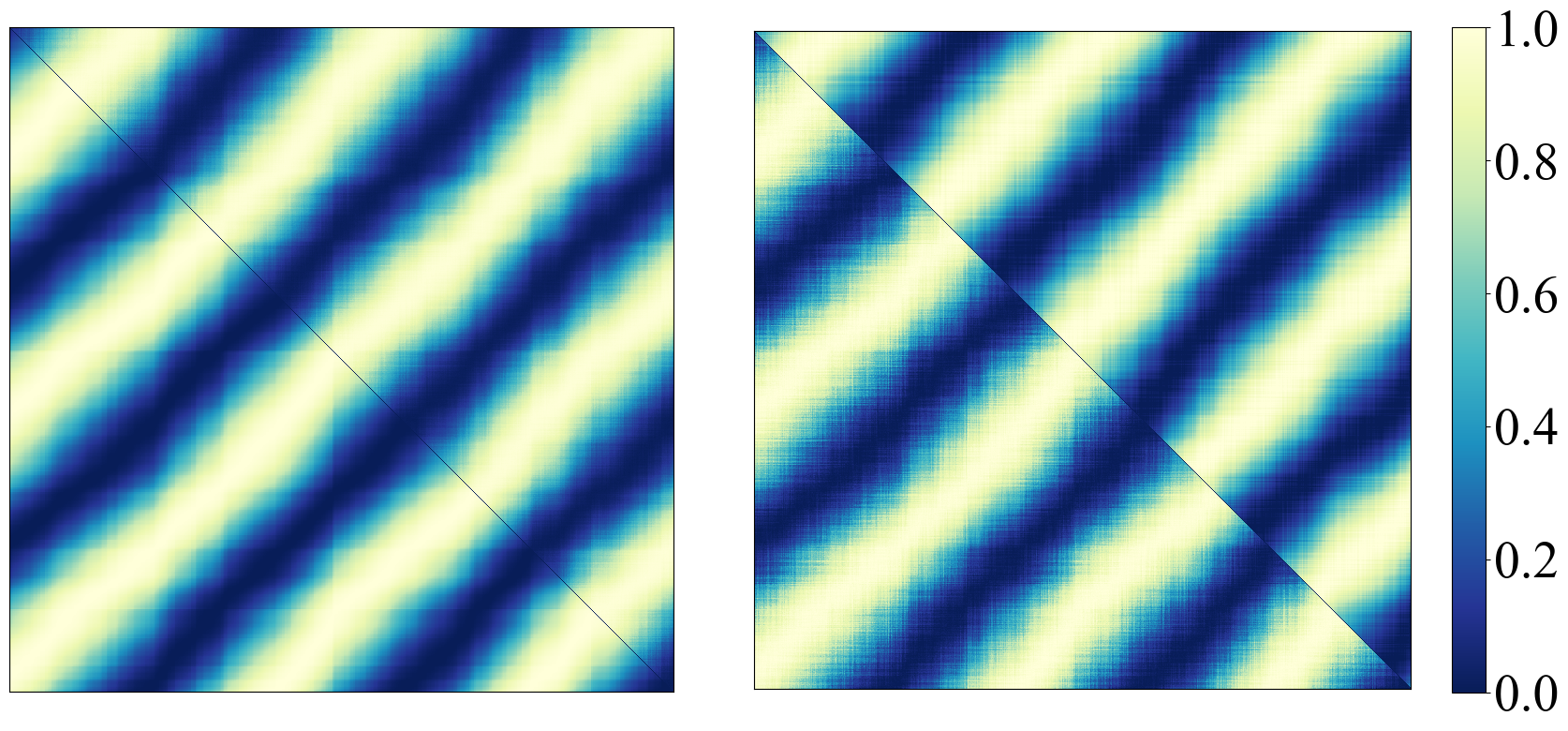} \\
	
	\vspace{0.5em} 
	
	\includegraphics[width=0.47\textwidth]{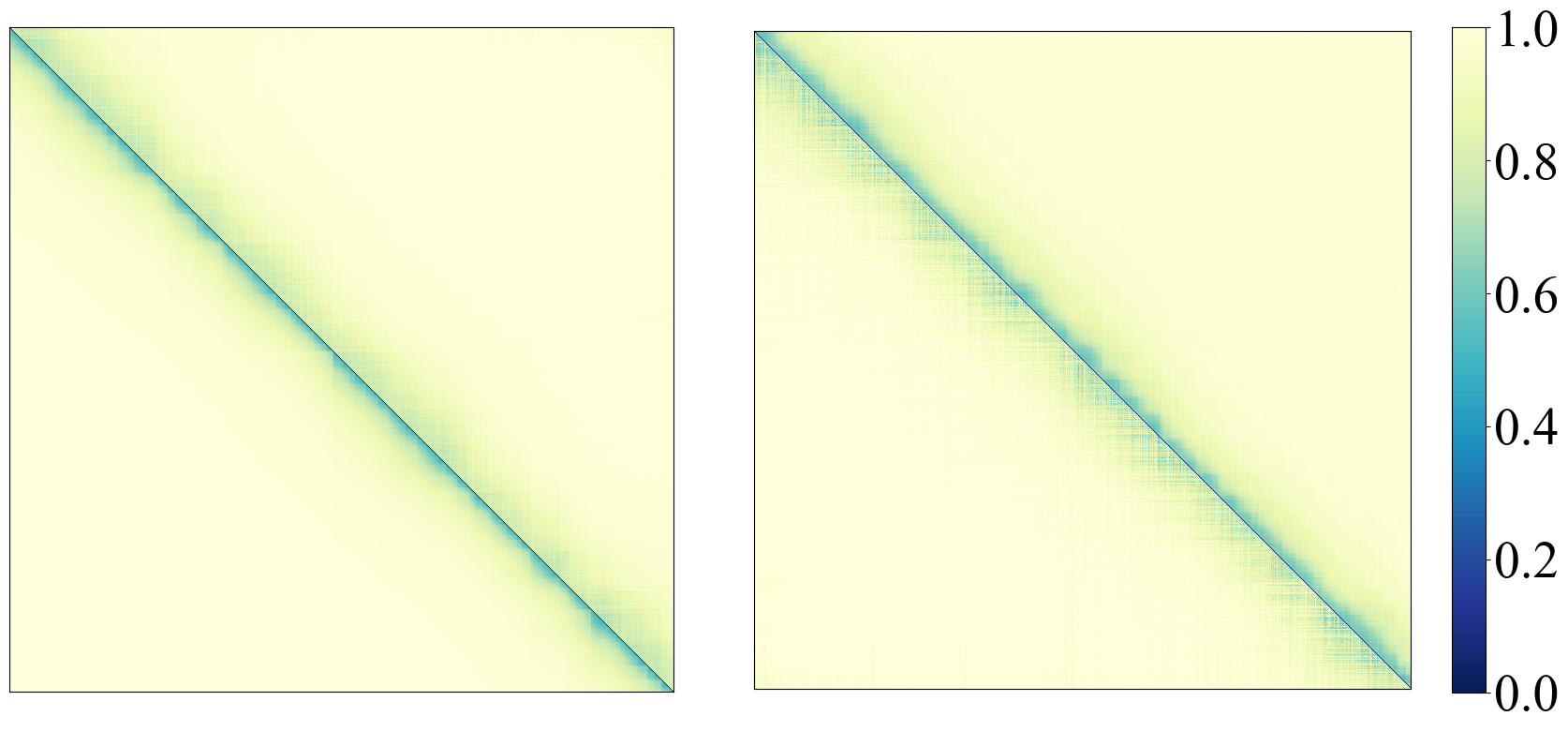} 
	\hspace{1em} 
	\includegraphics[width=0.47\textwidth]{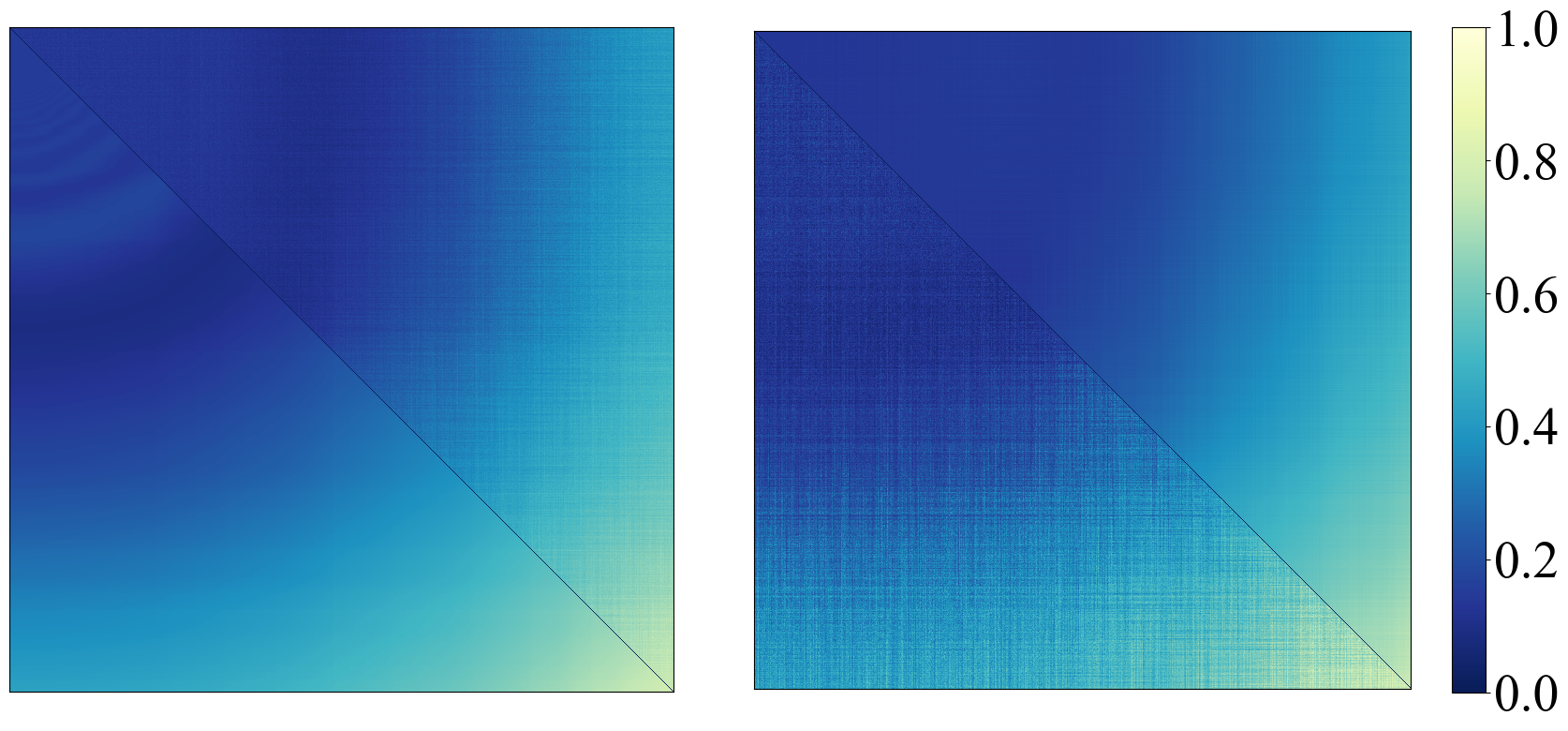} \\
	
	\caption{Estimated probability matrices for graphons 1–4, shown in 4 groups. For each group of images, The heatmap in the lower left corner shows the true $P$ values, while the one in the upper left displays estimates from our MNS method. The heatmaps on the right present estimates from the NS (lower triangle) method and the MultiNeSS method (upper triangle).}
	\hspace{1em} 
	\label{figure::heatmap}
\end{figure}

	
	
	

We fixed the number of network layers at 100 and compared the performances under difference node sizes with $n = 100, 200, 500, $ and 1000 to evaluate the impact of network scale on model fitting. For the choice of the constant factor $C$ in the bandwidth $h$ in Theorem \ref{theorem1}, we set $C =2$  for the rest of this paper, additional empirical results in the Supplementary Material show that our method is robust to the choice of the constant factor $C$. We compare the proposed method MNS with the neighborhood smoothing method for single-layer network (NS) proposed by \citet{zhang2017estimating} and the multiplex networks with shared structure model (MultiNeSS) proposed by \citet{macdonald2022latent}. For each method, we report the mean of root mean squared errors (RMSE) and mean absolute errors (MAE) with standard deviation in the bracket in  Table \ref{table::result}. All results are based on 200 replications. We present the heatmaps of the results for a single realization of the multi-layer network comprising 1000 nodes and 100 layers, specifically selecting the middle 50th layer for detailed analysis in Figure \ref{figure::heatmap}. The heatmaps for multi-layer networks at other representative positions are provided in the supplementary materials.

\begin{table}[htbp]
	\begin{center}
	\caption{Root mean squared errors and mean absolute errors with standard error (multiplied by 100). Each configuration is repeated 200 times for each multi-layer network.}
	\scalebox{0.95}{
		\begin{tabular}{ccccccc}
			\toprule
			\multicolumn{1}{p{4.235em}}{Node Size} &       & \multicolumn{1}{c}{} & 100   & 200   & 500   & 1000  \\
			\midrule
			\multicolumn{1}{c}{\multirow{6}[0]{*}{Graphon1}} & \multicolumn{1}{c}{\multirow{3}[0]{*}{RMSE}} & NS    & 8.46 (0.22)  & 5.93 (0.12)  & 3.75 (0.03) & 2.79 (0.01)  \\
			&       & MultiNeSS    &3.45 (0.28)	 &2.91 (0.20)	&2.36 (0.14)	&2.00 (0.10)  \\
			&       & \textcolor[rgb]{ 0,  0,  1}{MNS} & \textcolor[rgb]{ 0,  0,  1}{3.36 (0.13)} & \textcolor[rgb]{ 0,  0,  1}{2.56 (0.09)} & \textcolor[rgb]{ 0,  0,  1}{1.71 (0.04)} & \textcolor[rgb]{ 0,  0,  1}{1.19 (0.02)} \\
			\cmidrule{2-7}
			& \multicolumn{1}{c}{\multirow{3}[0]{*}{MAE}} & NS    & 5.88 (0.14) & 3.99 (0.08)  & 2.50 (0.02) & 1.90 (0.01)  \\
			&       & MultiNeSS    &2.07 (0.20)	&1.68 (0.14)	&1.39 (0.11)	&1.25 (0.10)\\
			&       & \textcolor[rgb]{ 0,  0,  1}{MNS} & \textcolor[rgb]{ 0,  0,  1}{2.29 (0.07)} & \textcolor[rgb]{ 0,  0,  1}{1.66 (0.05)} & \textcolor[rgb]{ 0,  0,  1}{1.07 (0.02)} & \textcolor[rgb]{ 0,  0,  1}{0.77 (0.01)} \\
			\midrule
			\multicolumn{1}{c}{\multirow{6}[0]{*}{Graphon2}} & \multicolumn{1}{c}{\multirow{3}[0]{*}{RMSE}} & NS    & 9.26 (0.22)  & 6.97 (0.12)  & 4.93 (0.06)  & 3.84 (0.04) \\
			&       & MultiNeSS    &7.17 (0.06)	&5.04 (0.05)	&3.18 (0.04)	&2.25 (0.01)\\
			&       & \textcolor[rgb]{ 0,  0,  1}{MNS} & \textcolor[rgb]{ 0,  0,  1}{5.35 (0.28)} & \textcolor[rgb]{ 0,  0,  1}{4.09 (0.20)} & \textcolor[rgb]{ 0,  0,  1}{2.99 (0.14)} & \textcolor[rgb]{ 0,  0,  1}{2.48 (0.12)} \\
			\cmidrule{2-7}
			& \multicolumn{1}{c}{\multirow{3}[0]{*}{MAE}} & NS    & 7.04 (0.15)  & 5.25 (0.08)  & 3.66 (0.04)  & 2.83 (0.03)  \\
			&       & MultiNeSS    &5.62 (0.05)	&3.96 (0.04)	&2.49 (0.03)	&1.76 (0.01)\\
			&       & \textcolor[rgb]{ 0,  0,  1}{MNS} & \textcolor[rgb]{ 0,  0,  1}{4.52 (0.20)} & \textcolor[rgb]{ 0,  0,  1}{3.44 (0.15)} & \textcolor[rgb]{ 0,  0,  1}{2.50 (0.11)} & \textcolor[rgb]{ 0,  0,  1}{2.07 (0.10)} \\
			\midrule
			\multicolumn{1}{c}{\multirow{6}[0]{*}{Graphon3}} & \multicolumn{1}{c}{\multirow{3}[0]{*}{RMSE}} & NS    & 8.93 (0.07)  & 6.89 (0.06)  & 5.04 (0.04)  & 3.99 (0.02)  \\
			&       & MultiNeSS    &6.76 (0.41)	&5.69 (0.31)	&4.57 (0.21)	&3.83 (0.19)\\
			&       & \textcolor[rgb]{ 0,  0,  1}{MNS} & \textcolor[rgb]{ 0,  0,  1}{4.92 (0.10)} & \textcolor[rgb]{ 0,  0,  1}{4.61 (0.05)} & \textcolor[rgb]{ 0,  0,  1}{3.86 (0.03)} & \textcolor[rgb]{ 0,  0,  1}{3.24 (0.03)} \\
			\cmidrule{2-7}
			& \multicolumn{1}{c}{\multirow{3}[0]{*}{MAE}} & NS    & 4.19 (0.11)  & 3.23 (0.08)  & 2.26 (0.05)  & 1.70 (0.03)  \\
			&       & MultiNeSS    &3.83 (0.32)	&3.49 (0.26)	&3.01 (0.17)	&2.59 (0.14)\\
			&       & \textcolor[rgb]{ 0,  0,  1}{MNS} & \textcolor[rgb]{ 0,  0,  1}{2.21 (0.07)} & \textcolor[rgb]{ 0,  0,  1}{1.90 (0.04)} & \textcolor[rgb]{ 0,  0,  1}{1.44 (0.02)} & \textcolor[rgb]{ 0,  0,  1}{1.13 (0.01)} \\
			\midrule
			\multicolumn{1}{c}{\multirow{6}[0]{*}{Graphon4}} & \multicolumn{1}{c}{\multirow{3}[0]{*}{RMSE}} & NS    & 10.43 (0.08)  & 7.97 (0.04)  & 5.65(0.02)   & 4.41 (0.01)  \\
			&       & MultiNeSS    &5.35 (0.32)	&4.58 (0.20)	&3.48 (0.11)	&2.79 (0.10)\\
			&       & \textcolor[rgb]{ 0,  0,  1}{MNS} & \textcolor[rgb]{ 0,  0,  1}{3.98 (0.13)} & \textcolor[rgb]{ 0,  0,  1}{3.28 (0.09)} & \textcolor[rgb]{ 0,  0,  1}{2.54 (0.06)} & \textcolor[rgb]{ 0,  0,  1}{2.06 (0.05)} \\
			\cmidrule{2-7}
			& \multicolumn{1}{c}{\multirow{3}[0]{*}{MAE}} & NS    & 8.29 (0.07)  & 6.34 (0.03)  & 4.49 (0.01)  & 3.50 (0.01)  \\
			&       & MultiNeSS    &4.61 (0.27)	&3.80 (0.18)	&2.83 (0.11)	&2.23 (0.04)\\
			&       & \textcolor[rgb]{ 0,  0,  1}{MNS} & \textcolor[rgb]{ 0,  0,  1}{3.21 (0.10)} & \textcolor[rgb]{ 0,  0,  1}{2.64 (0.08)} & \textcolor[rgb]{ 0,  0,  1}{2.04 (0.04)} & \textcolor[rgb]{ 0,  0,  1}{1.65 (0.03)} \\
			\bottomrule
		\end{tabular}
	}\label{table::result}
\end{center}
\end{table}%

Table \ref{table::result}  shows that the proposed MNS method has significant advantages in terms of estimating the probability matrices in all given network structures. Increasing the node sizes significantly lowers the estimation errors by incorporating  similarity information across layers.

We also conducted a series of simulations to evaluate the performance of the proposed method for different  size of layers. We fixed the number of nodes $n =500$ and varying the number of layers with $K = 10,20,50,100,200,\ldots,600$. Each configuration was repeated 100 times to ensure robust results. The results in Figure\ref{figure::simulation2} indicate that when the number of layers $K\leq10$, the proposed MNS method performs less effectively than the NS and MultiNesSS methods. However, after 20 layers, the error corresponding to the MNS method begins to fall below that of the other two methods and gradually decreases as the number of layers increases. This improvement stems from the richer neighborhood information provided by additional network layers, enabling the method to capture underlying network patterns more effectively. This also validates our theoretical assertion that when $K\gtrsim\left( n/\log n \right) ^{1/2}$, the incorporation of auxiliary information from similar layers within the multi-layer network would enhance the estimation accuracy.

\begin{figure}[htpb]
	\centering
	\includegraphics[width=0.8\textwidth,height=0.8\textwidth]{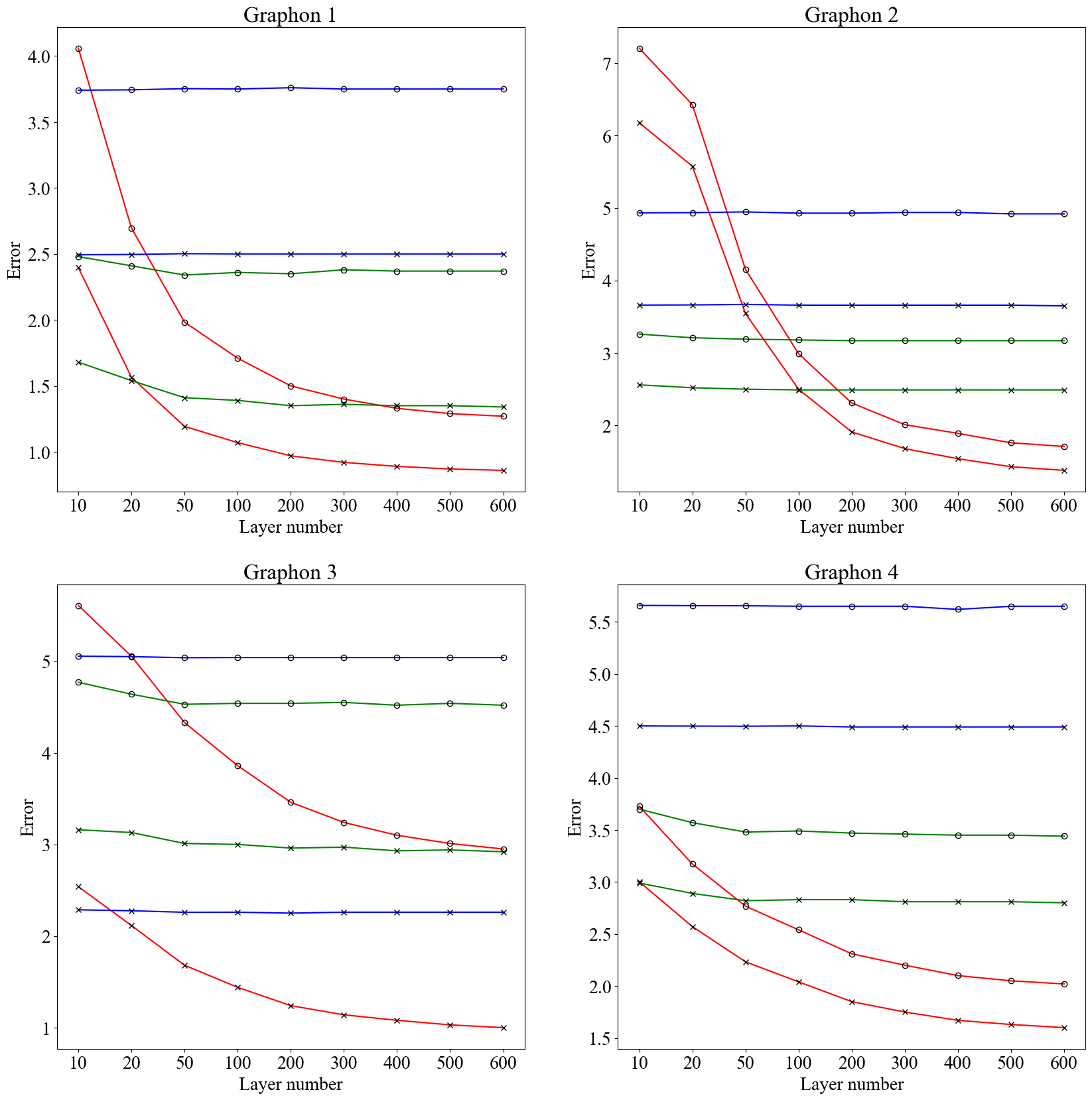} 
	\caption{Estimation errors (RMSE and MAE) for probability matrix. All results are multiplied by 100. Red line is for MNS method, green line is for MultiNeSS method and blue line is for NS. Lines marked with circles represent the changes in RMSE and lines marked with cross represent the changes in MAE.}
	\label{figure::simulation2}
\end{figure}

Finally, we fixed a small number of nodes $n=20$ and selected several larger layer depths $K$ ($K\gtrsim n^2\log n$). For each setting, we ran our MNS method and calculated the average RMSE and MAE after 100 experiments and plotted the results in the Figure \ref{fig:biglayer}. This experimental section is conducted because, as explained in Theorem \ref{theorem1}, when $K\gtrsim n^2\log n$, our convergence rate reaches $O(n^{c_2-1})$. This indicates that the error no longer decreases with increasing layer depth $K$. Our experiments also reveal that at very high layer depths, the error fluctuates only slightly within experimental error margins as the number of layers increases, further validating the theoretical correctness.

\begin{figure}[htbp]
    \centering
    \begin{minipage}{0.45\textwidth}
        \centering
        \includegraphics[width=\textwidth]{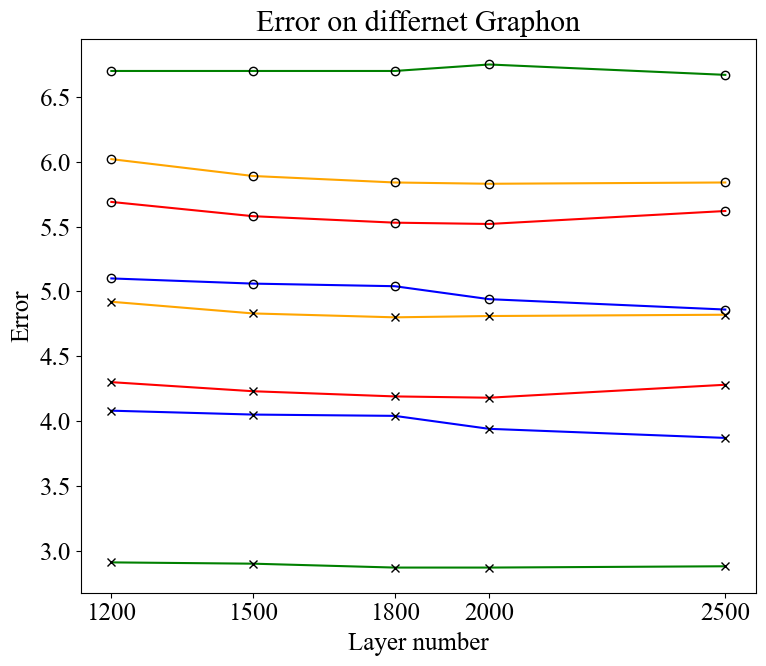}
        \caption{Estimation errors (RMSE and MAE) for probability matrix of MNS method. All results are multiplied by 100. Red line is for Graphon 1; blue line is for Graphon 2; green line is for Graphon 3 and orange line is for Graphon 4. Lines marked with circles represent the changes in RMSE and Lines marked with cross represent the changes in MAE.}
        \label{fig:biglayer}
    \end{minipage}
    \hfill
    \begin{minipage}{0.45\textwidth}
        \centering
        \includegraphics[width=\textwidth]{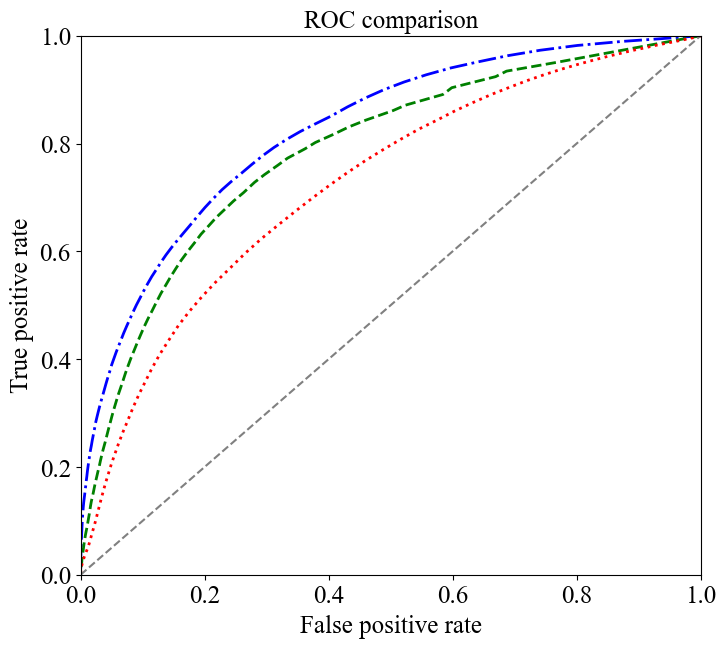}
        \caption{Receiver operating characteristic curves for link prediction on the overall FAO network; 10\% of edges are missing at random. The blue dot-dash curve is for our method (MNS), the  green dashed curve for NS and the red dotted line for MultiNeSS. The AUC values for the MNS, NS and MultiNeSS methods were 0.82, 0.78 and 0.72 respectively.}
        \label{fig:ROC}
    \end{minipage}
\end{figure}

\subsection{Real-world data analysis}

In this section, we apply the proposed MNS method to analyze the 2010 Worldwide Food Import/Export Network dataset \citep{dataset2010} from the Food and Agriculture Organization of the United Nations (\href{http://www.fao.org}{http://www.fao.org}). The dataset contains 364 networks among 214 countries with a total of 318,346 edges, where each network captures the trading connections of a specific food product among countries. 

We first preprocessed the data to select the nodes corresponding to the trading countries that are most relevant. Note that several major countries dominate the world economy activity and account for a substantial  number of trading connectivity, while the other countries with limited agricultural product categories have fewer trading relationships with each other  for  specific product networks. Therefore we focus our study on partial trading networks comprising of major countries whose corresponding degrees of nodes are larger than 9, which results in a subset of 51 countries that exert significant economic influence worldwide including the United States, mainland China, Japan and some European countries. 

For analyzing the trade networks of 364 commodities among these 52 countries, we further calculate the number of existing edges within each network, and observe that certain commodities exhibit relatively sparse transactions between major trading partners, leading to fewer observable edges, with many zero elements in those adjacency matrices. Notably, some networks have no edges, which make it impossible to identify the neighborhood sets for nodes. Such low observability in certain networks   make it extremely challenging for joint estimation of edge probabilities across multiple networks, leading to unreliable results. To address this issue, we selected the networks with a minimum of 300 non-zero elements in their adjacency matrices, resulting in a multi-layer network dataset comprising of 279 layers. This scheme ensures an adequate sample size and maintains a sufficiently large network scale.

Since the true probability matrix is unknown, it is challenging to directly examine and compare the results of different methods on real-world networks. Following the strategy proposed by \citet{zhang2017estimating}, we evaluate the empirical performance of different estimation methods through the link prediction task, which measures the practical utility of the estimated probability matrices. Specifically, we assume that $A$ is the true adjacency matrix and $A_{\mathrm{obs}}$ is the observed version, where each true edge is independently removed with probability $p$. Formally, this can be expressed as $A_{\mathrm{obs}} = M * A$, where $M_{ij}$ are independent Bernoulli$\left(1-p\right)$ random variables. The estimation $\hat{P}$ is then obtained based on $A_{\mathrm{obs}}$.

In contrast to \citet{zhang2017estimating}, which only considers single-layer networks, we evaluate link prediction performance in multi-layer settings. To assess the overall performance across all network layers, we generate corresponding $A_{obs}$ for each layer and define the false positive rate (FP) and true positive rate (TP) as follows:
\begin{equation}\label{equation::ROC2}
\begin{aligned}
\mathrm{FP(}t)=\sum_{i,j,k}{\mathbb{I}}\left( \hat{P}_{ijk}>t,A_{ijk}=0,M_{ijk}=0 \right) /\sum_{i,j,k}{\mathbb{I}}\left( A_{ijk}=0,M_{ijk}=0 \right),\\
\mathrm{TP(}t)=\sum_{i,j,k}{\mathbb{I}}\left( \hat{P}_{ijk}>t,A_{ijk}=1,M_{ijk}=0 \right) /\sum_{i,j,k}{\mathbb{I}}\left( A_{ijk}=1,M_{ijk}=0 \right).
\end{aligned}
\end{equation}
By varying the threshold $t$, we obtain the receiver operating characteristic (ROC) curve and compute the corresponding area under the curve (AUC), which summarizes the overall link prediction accuracy. In our experiments, 10\% of the edges are randomly removed from the FAO multi-layer trade networks, and the estimation procedures are repeated ten times to obtain averaged FP and TP rates.

Figure \ref{fig:ROC} presents the aggregated ROC curves across all layers, and the resulting AUC values are 0.82 for the proposed MNS method, 0.78 for NS method and 0.72 for the MultiNeSS method. The superior AUC of the MNS method demonstrates its advantage in capturing cross-layer dependencies, leading to more accurate probability estimation and improved link prediction performance on real-world multi-layer networks.

\section{Proofs}

	In this section, we provide the detailed proof of Theorem \ref{theorem1} and Theorem \ref{theorem2}. For convenience, we start with recalling notations and assumptions made in the main paper. Let $0=x_0<x_1<...<x_R=1$, $I_r:=\left[ x_{r-1},x_r \right) $ for $1\le r\le R-1$; $I_R=\left[ x_{R-1},x_R \right] $ and let $0=z_0<z_1<...<z_T=1$, $H_z:=\left[ z_{t-1},z_t \right)$ for $1\le z\le T-1$ ; $H_T=\left[ z_{T-1},z_T \right] $. Assume the graphon $f$ is the Multi-layer Lipschitz function on each of $I_r\times I_s\times H_t$ for $1\le r,s\le R$, $1\le t\le T$. Let $L_n$ and $L_K$ denote the maximum Lipschitz constant.
	\begin{assumption}
		The number of pieces $R$ and $T$ may grow with $n$ and $K$, as long as $\min_r \left| I_r \right|/\left\{ \left( nK \right) ^{-1}\log n \right\} ^{1/3}$, $\min_t \left| H_t \right|/\left\{ \left( nK \right) ^{-1}\log n \right\} ^{1/3}\rightarrow \infty $.
	\end{assumption}
	
	For any $\xi \in [0,1]$, let $I(\xi )$ denote the $I_r$ that contains $\xi $. Let $S_i(\Delta )=\left[ \xi _i-\Delta ,\xi _i+\Delta \right] \cap I\left( \xi _i \right) $ denote the neighborhood of $\xi _i$ in which $f(x,y,z)$ is Lipschitz in $x\in S_i(\Delta )$ for any fixed $y$, $z$,  or $y\in S_i(\Delta )$ for any fixed $x$, $z$. For any $\eta \in [0,1]$, let $H(\eta )$ denote the $H_z$ that contains $\eta $. Let $S_k(\Delta )=\left[ \eta _k-\Delta ,\eta _k+\Delta \right] \cap H\left( \eta _k \right) $ denote the neighborhood of $\eta _k$ in which $f(x,y,z)$ is Lipschitz in $z\in S_k(\Delta )$ for any fixed $x$ and $y$.

We first consider the case where $K\ge C^{-3/2}\sqrt{\frac{n}{\log n}}$ so that $C\left\{ \left( nK \right) ^{-1}\log n \right\} ^{1/3}\ge 1/K$. Suppose we select the neighborhood $\mathcal{N}_{i}^{k}$ of node $i$ in the $k$th network by thresholding at the lower $h_1$th quantile of $\{d(i,k)\}_{k\ne i}$ and  select the neighborhood of the $k$th network $\mathcal{N} ^{k}$ by thresholding at the lower $h_2$th quantile of $\{d(i,k)\}_{k\ne i}$ respectively and set $h_1=h_2=C\left\{ \left( nK \right) ^{-1}\log n \right\} ^{1/3}$. 
Recall our estimator is defined as
$$
\tilde{P}_{ijk}=\frac{\sum\nolimits_{k^{\prime}\in \mathcal{N} ^k}^{}{\sum\nolimits_{i^{\prime}\in \mathcal{N} _{i}^{k^{\prime}}}^{}{A_{i^{\prime}jk^{\prime}}}}}{\sum\nolimits_{k^{\prime}\in \mathcal{N} ^k}^{}{\left| \mathcal{N} _{i}^{k^{\prime}} \right|}}.
$$

We denote $M=\sum\nolimits_{k^{\prime}\in \mathcal{N} ^k}^{}{\left| \mathcal{N} _{i}^{k^{\prime}} \right|}$ and  perform a bias-variance decomposition:
\begin{align}\label{decomposition}
	&\frac{1}{n^2}\sum\nolimits_{i,j}^{}{\left( \tilde{P}_{ijk}-P_{ijk} \right) ^2}=\frac{1}{n^2}\sum_{i,j}{\left\{ \frac{\sum\nolimits_{k^{\prime}\in \mathcal{N}^k}^{}{\sum\nolimits_{i^{\prime}\in \mathcal{N}_{i}^{k^{\prime}}}^{}{\left( A_{i^{\prime}jk^{\prime}}-P_{ijk} \right)}}}{\sum\nolimits_{k^{\prime}\in \mathcal{N}^k}^{}{\left| \mathcal{N}_{i}^{k^{\prime}} \right|}} \right\} ^2}    \nonumber   \\
	=&\frac{1}{n^2}\sum_{i,j}{\left[ \frac{\sum\nolimits_{k^{\prime}\in \mathcal{N}^k}^{}{\sum\nolimits_{i^{\prime}\in \mathcal{N}_{i}^{k^{\prime}}}^{}{\left\{ \left( A_{i^{\prime}jk^{\prime}}-P_{i^{\prime}jk^{\prime}} \right) +\left( P_{i^{\prime}jk^{\prime}}-P_{ijk^{\prime}} \right) +\left( P_{ijk^{\prime}}-P_{ijk} \right) \right\}}}}{\sum\nolimits_{k^{\prime}\in \mathcal{N}^k}^{}{\left| \mathcal{N}_{i}^{k^{\prime}} \right|}} \right] ^2}   \nonumber   \\
	\le &\frac{3}{n^2\left( \sum\nolimits_{k^{\prime}\in \mathcal{N}^k}^{}{\left| \mathcal{N}_{i}^{k^{\prime}} \right|} \right) ^2}\sum_{i,j}{\left\{ J_1\left( i,j,k \right) +J_2\left( i,j,k \right) +J_3\left( i,j,k \right) \right\}} 
	\nonumber  \\
	=&\frac{3}{n^2M^2}\sum_{i,j}{\left\{ J_1\left( i,j,k \right) +J_2\left( i,j,k \right) +J_3\left( i,j,k \right) \right\}}.
\end{align}

According to Cauchy-Schwarz inequality, we have $
J_1\left( i,j,k \right) =\left( \sum_{k^{\prime},i^{\prime}}{\left( A_{i^{\prime}jk^{\prime}}-P_{i^{\prime}jk^{\prime}} \right)} \right) ^2$,\\$J_2\left( i,j,k \right) =\left( \sum_{k^{\prime},i^{\prime}}{\left( P_{i^{\prime}jk^{\prime}}-P_{ijk^{\prime}} \right)} \right) ^2$and $J_3\left( i,j,k \right) =\left( \sum_{k^{\prime},i^{\prime}}{\left( P_{ijk^{\prime}}-P_{ijk} \right)} \right) ^2$,where $k^{\prime}\in \mathcal{N} ^k,i^{\prime}\in \mathcal{N} _{i}^{k^{\prime}}$. Thus, our goal is to bound \eqref{decomposition}.

Next we give two  Lemmas which are critical for our further theoretical analysis to $J_2\left( i,j,k \right)$ and $J_3\left( i,j,k \right)$.
Recall that we used a measure of closeness of adjacency matrix slices for nodes as 
\begin{equation*}
	\tilde{d}^2\left(i, i^{\prime}\right)=\max _{k \neq i, i^{\prime}}\left|\left\langle A_i-A_{i^{\prime}}, A_{k \cdot}\right\rangle\right| / n=\max _{k \neq i, i^{\prime}}\left|\left(A^2 / n\right)_{i k}-\left(A^2 / n\right)_{i^{\prime} k}\right|.
\end{equation*} 
We define a new distance measure in order to proof conveniently.
\begin{equation*}
	\tilde{D}\left( i,i^{\prime} \right) =\frac{1}{\left| \mathcal{N} ^k \right|}\sum\nolimits_{k^{\prime}\in \mathcal{N} ^k}^{}{\tilde{d}\left( i,i^{\prime} \right) =\frac{1}{n\left| \mathcal{N} ^k \right|}\underset{j\ne i,i^{\prime}}{\max}\left| \sum_{\ell ,k^{\prime}\in \mathcal{N} ^k}{\left( A_{i\ell k^{\prime}}A_{\ell jk^{\prime}}-A_{i^{\prime}\ell k^{\prime}}A_{\ell jk^{\prime}} \right)} \right|}.
\end{equation*} 
\begin{lemma}\label{lemma1}
	Let $\tilde{C_0},\tilde{C_2}>0$ be arbitrary global constants and assume $n$ and $K$ are large enough, so that 
	\[
	\left\{ \left( \tilde{C_2}+8 \right) \left( nK \right) ^{-1}\log n \right\} ^{1/3} \le 1, \text{and } K\ge C^{-3/2}\sqrt{\frac{n}{\log n}}.
	\]
	Then, 
	\begin{enumerate}
		\item With probability $1-2n^{-\tilde{C_2}/4}$, for all $i$ and $i^{\prime}\in \mathcal{N} _i$, we have
		\begin{equation*}
			\sum_{j,k^{\prime}\in N^k}{\left( P_{i^{\prime}jk^{\prime}}-P_{ijk^{\prime}} \right)}\le n\left| \mathcal{N} ^k \right|\left[ \left\{ 6\tilde{C}_0L_n+8\left( \tilde{C_2}+8 \right) ^{1/3} \right\} \left( \frac{\log n}{nK} \right) ^{1/3}+\frac{32}{n} \right] .
		\end{equation*}
		\item With probability $1-2n^{-\tilde{C_2}/4}$, for all $i$ and all $i^{\prime}\in \mathcal{N} _i$ simultaneously, we have
		\begin{equation*}
			\tilde{D}\left( i,i^{\prime} \right) \le \left\{ \tilde{C}_0L_n+2\left( \tilde{C_2}+8 \right) ^{1/3} \right\} \left( \frac{\log n}{nK} \right) ^{1/3}+\frac{8}{n}.
		\end{equation*}      
	\end{enumerate}
\end{lemma}
\begin{proof}[Proof of Lemma \ref{lemma1}]
	We start with concentration results. For any $i,j$ such that $i\ne j$, we have\\
	\begin{align}\label{lemma1_1}
		&\frac{1}{n\left| \mathcal{N} ^k \right|}\left| \sum\nolimits_{\ell,k^{\prime}}^{}{\left( P_{i\ell k^{\prime}}P_{j\ell k^{\prime}}-A_{i\ell k^{\prime}}A_{j\ell k^{\prime}} \right)} \right|   \nonumber \\
		\le &\frac{\left| \sum\nolimits_{\ell \ne i,j;k^{\prime}}^{}{\left( P_{i\ell k^{\prime}}P_{j\ell k^{\prime}}-A_{i\ell k^{\prime}}A_{j\ell k^{\prime}} \right)} \right|}{\left( n-2 \right) \left| \mathcal{N} ^k \right|}\cdot \frac{n-2}{n}+\frac{\sum\nolimits_{k^{\prime}\in \mathcal{N} ^k}^{}{\left| \left( A_{iik^{\prime}}+A_{jjk^{\prime}} \right) A_{ijk^{\prime}} \right|}}{n\left| \mathcal{N} ^k \right|}   \nonumber \\
		&+\frac{\sum\nolimits_{k^{\prime}\in \mathcal{N} ^k}^{}{\left| \left( P_{iik^{\prime}}+P_{jjk^{\prime}} \right) P_{ijk^{\prime}} \right|}}{n\left| \mathcal{N} ^k \right|}   \nonumber  \\
		\le &\frac{\left| \sum\nolimits_{\ell \ne i,j;k^{\prime}}^{}{\left( P_{i\ell k^{\prime}}P_{j\ell k^{\prime}}-A_{i\ell k^{\prime}}A_{j\ell k^{\prime}} \right)} \right|}{\left( n-2 \right) \left| \mathcal{N} ^k \right|}\cdot \frac{n-2}{n}+\frac{4}{n}
	\end{align}
	By Bernstein's inequality, for any $ 0< \epsilon \le 1 $ and $n\ge3$, we have
	\begin{align*}
		&P\left\{ \frac{\left| \sum\nolimits_{\ell \ne i,j;k^{\prime}}^{}{\left( P_{i\ell k^{\prime}}P_{j\ell k^{\prime}}-A_{i\ell k^{\prime}}A_{j\ell k^{\prime}} \right)} \right|}{\left( n-2 \right) \left| \mathcal{N} ^k \right|}\ge \epsilon \right\}\\
		\le &2\exp \left\{ -\frac{\left( n-2 \right) \left| \mathcal{N} ^k \right|\epsilon ^2}{2\left( 1+\epsilon /3 \right)} \right\} \le 2\exp \left\{ -\frac{1}{4}n\left| \mathcal{N} ^k \right|\epsilon ^2 \right\}.
	\end{align*}
	Taking a union bound over all $i\ne j$, we have
	\begin{equation*}
		P\left\{ \underset{i,j:i\ne j}{\max}\frac{\left| \sum\nolimits_{\ell \ne i,j;k^{\prime}}{\left( P_{i\ell k^{\prime}}P_{j\ell k^{\prime}}-A_{i\ell k^{\prime}}A_{j\ell k^{\prime}} \right)} \right|}{\left( n-2 \right) \left| \mathcal{N} ^k \right|}\ge \epsilon \right\} \le 2n^2\exp \left\{ -\frac{1}{4}n\left| \mathcal{N} ^k \right|\epsilon ^2 \right\}.
	\end{equation*}
	Then setting $\epsilon =\left\{ \left( \tilde{C_2}+8 \right)^{2/3} \left( n\left| \mathcal{N} ^k \right| \right) ^{-1}\log n \right\} ^{1/2}=\left\{ \left( \tilde{C_2}+8 \right) \left( nK \right) ^{-1}\log n \right\} ^{1/3}$, we have
	\begin{equation}\label{lemma1_2}
		P\left\{ \underset{i,j:i\ne j}{\max}\frac{\left| \sum\nolimits_{l\ne i,j;k^{\prime}}{\left( P_{ilk^{\prime}}P_{jlk^{\prime}}-A_{ilk^{\prime}}A_{jlk^{\prime}} \right)} \right|}{\left( n-2 \right) \left| \mathcal{N} ^k \right|}\ge \left\{ \frac{\left( \tilde{C_2}+8 \right) \log n}{nK} \right\} ^{1/3} \right\} \le 2n^{-\tilde{C_2}/4}
	\end{equation}
	Combining \eqref{lemma1_1} and \eqref{lemma1_2}, with probability $1- 2n^{-\tilde{C_2}/4}$, for $n$ enough the following holds
	\begin{equation}\label{lemma1_3}
		\frac{\left| \sum\nolimits_{l,k^{\prime}}^{}{\left( P_{ilk^{\prime}}P_{jlk^{\prime}}-A_{ilk^{\prime}}A_{jlk^{\prime}} \right)} \right|}{n\left| \mathcal{N} ^k \right|}\le \left\{ \frac{\left( \tilde{C_2}+8 \right) \log n}{nK} \right\} ^{1/3}+\frac{4}{n}
	\end{equation}
	
	Next we show the properties of Lipschitz in the same layer network and the similar elements,
	\begin{align}\label{lemma1_4}
		\frac{1}{n\left| \mathcal{N} ^k \right|}\left| \sum_{j,k^{\prime}\in \mathcal{N} ^k}{\left( P_{i^{\prime}jk^{\prime}}^{2}-P_{i^{\prime}jk^{\prime}}P_{\tilde{i}^{\prime}jk^{\prime}} \right)} \right|&\le \frac{1}{n\left| \mathcal{N} ^k \right|}\sum_{j,k^{\prime}}{\left| P_{i^{\prime}jk^{\prime}} \right|\cdot \left| P_{i^{\prime}jk^{\prime}}-P_{\tilde{i}^{\prime}jk^{\prime}} \right|}   \nonumber \\
		&\le \frac{1}{n\left| \mathcal{N} ^k \right|}\sum_{j,k^{\prime}}{\left| P_{i^{\prime}jk^{\prime}}-P_{\tilde{i}^{\prime}jk^{\prime}} \right|}\le L_n\Delta _n,
	\end{align}
	where the last inequality follows from
	\begin{equation*}
		\left| P_{i^{\prime}\ell k}-P_{i\ell k} \right|=\left| f\left( \xi _{i^{\prime}},\xi _{\ell},\eta _k \right) -f\left( \xi _i,\xi _{\ell},\eta _k \right) \right|\le L_n\left| \xi _{i^{\prime}}-\xi _i \right|\le L_n\Delta _n
	\end{equation*}
	for all $\ell$ and $k$.
	
	We are now ready to upper bound $\tilde{D}\left( i,i^{\prime} \right) $ for $i^{\prime}\in \mathcal{N} _{i}^{k^{\prime}}$. We bound $\tilde{D}\left( i,i^{\prime} \right)$ via bounding $\tilde{D}(i,\tilde{i})$ for $\tilde{i}$ with $\xi _{\tilde{i}}\in S_i\left( \Delta _n \right)$,which is equivalent to bound the $\tilde{d}\left( i,i^{\prime} \right)$ by $\tilde{d}\left( i,\tilde{i} \right)$. By \eqref{lemma1_3} and \eqref{lemma1_4}, with probability $1-2n^{-\tilde{C_2}/4}$, we have
	\begin{align}\label{dii}
		&\tilde{D}\left( i,\tilde{i} \right) =\underset{j\ne i,\tilde{i}}{\max}\frac{1}{n\left| \mathcal{N} ^k \right|}\left| \sum_{\ell ,k^{\prime}\in \mathcal{N} ^k}{\left( A_{i\ell k^{\prime}}A_{\ell jk^{\prime}}-A_{\tilde{i}\ell k^{\prime}}A_{\ell jk^{\prime}} \right)} \right| \nonumber \\
		\le &\underset{j\ne i,\tilde{i}}{\max}\frac{1}{n\left| \mathcal{N} ^k \right|}\left| \sum_{\ell ,k^{\prime}\in \mathcal{N} ^k}{\left( P_{i\ell k^{\prime}}P_{\ell jk^{\prime}}-P_{\tilde{i}\ell k^{\prime}}P_{\ell jk^{\prime}} \right)} \right|+\frac{1}{n\left| \mathcal{N} ^k \right|}\left| \sum_{\ell ,k^{\prime}\in \mathcal{N} ^k}{\left( A_{i\ell k^{\prime}}A_{\ell jk^{\prime}}-P_{i\ell k^{\prime}}P_{\ell jk^{\prime}} \right)} \right| \nonumber \\
		&+\frac{1}{n\left| \mathcal{N} ^k \right|}\left| \sum_{\ell ,k^{\prime}\in \mathcal{N} ^k}{\left( A_{\tilde{i}\ell k^{\prime}}A_{\ell jk^{\prime}}-P_{\tilde{i}\ell k^{\prime}}P_{\ell jk^{\prime}} \right)} \right| \nonumber \\
		\le &L_n\Delta _n+2\left[ \left\{ \frac{\left( \tilde{C_2}+8 \right) \log n}{nK} \right\} ^{1/3}+\frac{4}{n} \right].
	\end{align}
	Now since the fraction of nodes contained in $\left| \left\{ \tilde{i}:\xi _{\tilde{i}}\in S_i\left( \Delta _n \right) \right\} \right|$ is at least $h=C\left\{ \left( nK \right) ^{-1}\log n \right\} ^{1/3}$, this upper bounds $\tilde{D}\left( i,i^{\prime} \right)$ for $i^{\prime}\in \mathcal{N} _{i}^{k^{\prime}}$, since nodes in $i^{\prime}\in \mathcal{N} _{i}^{k^{\prime}}$ have the lowest $h$ fraction of values in $\{\tilde{d}(i,k)\}_k$. Setting $\Delta _n$ as $\tilde{C}_0\left\{ \left( nK \right) ^{-1}\log n \right\} ^{1/3}$ and $0<C \le \tilde{C}_0$, by \eqref{lemma1_3}, with probability $1-2n^{-\tilde{C_2}/4}$, for all $i$, at least $\tilde{C}_0\left\{ \left( nK \right) ^{-1}\log n \right\} ^{1/3}$ fraction of nodes $\tilde{i}\ne i$ satisfy both $\xi _{\tilde{i}}\in S_i\left( \Delta _n \right) $ and \eqref{dii}.
	Then we yields that
	\begin{align}\label{lemma1_5}
		\tilde{D}\left( i,i^{\prime} \right) &\le L_n\Delta _n+2\left\{ \frac{\left( \tilde{C_2}+8 \right) \log n}{nK} \right\} ^{1/3}+\frac{8}{n}\nonumber \\
		&=\left\{ \tilde{C}_0L_n+2\left( \tilde{C_2}+8 \right) ^{1/3} \right\} \left( \frac{\log n}{nK} \right) ^{1/3}+\frac{8}{n}.
	\end{align}
	holds for all $i$ and all $i^{\prime}\in \mathcal{N} _i$ simultaneously with growing $n$ .
	
	We are now ready to complete the proof of the first claim of Lemma \ref{lemma1}. By \eqref{lemma1_3}, \eqref{lemma1_4}, \eqref{lemma1_5}, with probability $1-2n^{-\tilde{C_2}/4}$, the following holds. For $n$ large enough and for all $i$ and $i^{\prime} \in \mathcal{N}_i$ we can find $\tilde{i} \in S_i\left(\Delta_n\right)$ and $\tilde{i^{\prime}} \in S_{i^{\prime}}\left(\Delta_n\right)$ such that $i, i^{\prime}, \tilde{i}$ and $\tilde{i^{\prime}}$ are different from each other. Then by adding and subtracting items, tackling the $\boldsymbol{A}$ by \eqref{lemma1_4}, the $\boldsymbol{B}$ by \eqref{lemma1_3}, the $\boldsymbol{C}$ by \eqref{lemma1_5}, finally we have
	\begin{equation*}
		\begin{split}
			&\sum_{k^{\prime}\in \mathcal{N} ^k,j}{\left( P_{i^{\prime}jk^{\prime}}-P_{ijk^{\prime}} \right)}^{2}\le \left| \sum_{k^{\prime}\in \mathcal{N} ^k,j}{\left( P_{i^{\prime}jk^{\prime}}^{2}-P_{i^{\prime}jk^{\prime}}P_{ijk^{\prime}} \right)} \right|+\left| \sum_{k^{\prime}\in \mathcal{N} ^k,j}{\left( P_{ijk^{\prime}}^{2}-P_{ijk^{\prime}}P_{i^{\prime}jk^{\prime}} \right)} \right|\\
			\le &\underset{\boldsymbol{A}}{\underbrace{\left| \sum_{j,k^{\prime}}{\left( P_{i^{\prime}}^{2}-P_{i^{\prime}}P_{\tilde{i}^{\prime}} \right)} \right|+\left| \sum_{j,k^{\prime}}{\left( P_{i^{\prime}}P_{\tilde{i}}-P_{i^{\prime}}P_i \right)} \right|+\left| \sum_{j,k^{\prime}}{\left( P_{i}^{2}-P_iP_{\tilde{i}} \right)} \right|+\left| \sum_{j,k^{\prime}}{\left( P_{\tilde{i}}P_{i^{\prime}}-P_iP_{i^{\prime}} \right)} \right|}}\\
			&+\left| \sum_{k^{\prime}\in \mathcal{N} ^k,j}{\left( P_{i^{\prime}}P_{\tilde{i}^{\prime}}-P_{i^{\prime}}P_{\tilde{i}} \right)} \right|+\left| \sum_{k^{\prime}\in \mathcal{N} ^k,j}{\left( P_iP_{\tilde{i}}-P_{\tilde{i}}P_{i^{\prime}} \right)} \right|\\
			\le &4n\left| \mathcal{N} ^k \right|L_n\Delta _n+\left| \sum_{k^{\prime}\in \mathcal{N} ^k,j}{\left( P_{i^{\prime}}P_{\tilde{i}^{\prime}}-P_{i^{\prime}}P_{\tilde{i}} \right)} \right|+\left| \sum_{k^{\prime}\in \mathcal{N} ^k,j}{\left( P_iP_{\tilde{i}}-P_{\tilde{i}}P_{i^{\prime}} \right)} \right|\\
			\le &\underset{\boldsymbol{B}}{\underbrace{\left| \sum_{j,k^{\prime}}{\left( P_{i^{\prime}}P_{\tilde{i}^{\prime}}-A_{i^{\prime}}A_{\tilde{i}^{\prime}} \right)} \right|+\left| \sum_{j,k^{\prime}}{\left( P_{i^{\prime}}P_{\tilde{i}}-A_{i^{\prime}}A_{\tilde{i}} \right)} \right|+\left| \sum_{j,k^{\prime}}{\left( P_iP_{\tilde{i}}-A_iA_{\tilde{i}} \right)} \right|+\left| \sum_{j,k^{\prime}}{\left( P_{\tilde{i}}P_{i^{\prime}}-A_{\tilde{i}}A_{i^{\prime}} \right)} \right|}}\\
			&+\underset{\boldsymbol{C}}{\underbrace{\left| \sum_{j,k^{\prime}}{\left( A_{i^{\prime}}A_{\tilde{i}^{\prime}}-A_{i^{\prime}}A_{\tilde{i}} \right)} \right|+\left| \sum_{j,k^{\prime}}{\left( A_iA_{\tilde{i}}-A_{\tilde{i}}A_{i^{\prime}} \right)} \right|}}+4n\left| N^k \right|L_n\Delta _n\\
			\le &4n|\mathcal{N} ^k|\left[ \left\{ \frac{\left( \tilde{C_2}+8 \right) \log n}{nK} \right\} ^{1/3}+\frac{4}{n} \right] +2n|\mathcal{N} ^k|\tilde{D}\left( i,i^{\prime} \right) +4n\left| \mathcal{N} ^k \right|L_n\Delta _n\\
			=&n\left| \mathcal{N} ^k \right|\left[ \left\{ 6\tilde{C}_0L_n+8\left( \tilde{C_2}+8 \right) ^{1/3} \right\} \left( \frac{\log n}{nK} \right) ^{1/3}+\frac{32}{n} \right] .\\
		\end{split}
	\end{equation*}

	This completes the proof of Lemma \ref{lemma1}.
\end{proof}

Next we give the other Lemma, for all $k$ and $k^{\prime}\in \mathcal{N} _k$ we can find $\tilde{k}\in S_k\left( \Delta _K \right) $ and $\tilde{k}^{\prime}\in S_{k^{\prime}}\left( \Delta _K \right) $ such that $k$, $k^{\prime}$, $\tilde{k}$ and $\tilde{k}^{\prime}$ are different from each other. We denote $P_k$ as $ P_{ijk}$.

Recall that we defined a measure of closeness as 
\begin{equation*}
	\tilde{d}\left( k,k^{\prime} \right) =\max_{l\ne k,k^{\prime}} \left| \mathrm{tr}\left[ \left( A^k-A^{k^{\prime}} \right) ^{\mathrm{T}}A^l \right] \right|/n^2=\max_{l\ne k,k^{\prime}} \left| \mathrm{tr}\left( {A^k}^{\mathrm{T}}A^l \right) -\mathrm{tr}\left( {A^{k^{\prime}}}^{\mathrm{T}}A^l \right) \right|/n^2.
\end{equation*}

\begin{lemma}\label{lemma2}
	Let $0 < C_3 < 1$ be an arbitrary global constants chosen to be as small as possible and assume $n$ is large enough so that 
	\begin{enumerate}
		\item With probability $1-2\exp \left( -n^{C_3/4} \right)$, for all $i$ and $i^{\prime}\in \mathcal{N} _i$, we have
		\begin{equation*}
			\frac{1}{n^2}\left| \sum\nolimits_{i,j}^{}{\left( A_{ijq}A_{ijl}-P_{ijq}P_{ijl} \right)} \right|\le \frac{1}{n^{1-C_3}}.
		\end{equation*}
		\item With probability $1-2\exp \left( -n^{C_3/4} \right)$, for all $k$ and all $k^{\prime}\in \mathcal{N} _i$ simultaneously, we have
		\begin{equation*}
			\tilde{d}\left( k,k^{\prime} \right) \le L_K\Delta _K+\frac{1}{n^{1-C_3}}
		\end{equation*}		
	\end{enumerate}
\end{lemma}
\begin{proof}[Proof of Lemma \ref{lemma2}]
	We start with concentration results. For fixed $p$, $q$ such that $p \ne q$, By Bernstein's inequality, for any $0<\epsilon \le 1$ we have
	\begin{equation*}
		P\left\{ \left| \frac{\sum\nolimits_{i,j}^{}{\left( A_{ijq}A_{ijl}-P_{ijq}P_{ijl} \right)}}{n^2} \right|\ge \epsilon \right\} \le 2\exp \left\{ -\frac{\frac{1}{2}n^2\epsilon ^2}{1+\frac{1}{3}\epsilon} \right\} \le 2\exp \left( -\frac{1}{4}n^2\epsilon ^2 \right).
	\end{equation*}
	Then setting $\epsilon =1/n^{1-C_3}$ with $n$ large enough so that $\epsilon \le 1$, $0<C_3<1$ be an arbitrary global constants chosen to be as small as possible and let $n$ be large enough, we have
	\begin{equation}\label{lemma2_1}
		P\left\{ \left| \frac{\sum\nolimits_{i,j}^{}{\left( A_{ijq}A_{ijl}-P_{ijq}P_{ijl} \right)}}{n^2} \right|\ge \frac{1}{n^{1-C_3}} \right\} \le 2\exp \left( -n^{C_3/4} \right). 
	\end{equation}
	This completes the  claim 1 of Lemma \ref{lemma2}.
	
	Next we show an inequality which uses the property of Lipschitz.
	\begin{align}\label{lemma2_2}
		&\left| \sum\nolimits_{i,j}^{}{\left( P_{ijk}P_{ijk}-P_{ijk}P_{ij\tilde{k}} \right)} \right|=\left| \sum\nolimits_{i,j}^{}{P_{ijk}\left( P_{ijk}-P_{ij\tilde{k}} \right)} \right|\le \left| \sum\nolimits_{i,j}^{}{\left( P_{ijk}-P_{ij\tilde{k}} \right)} \right|   \nonumber \\
		\le& \sum\nolimits_{i,j}^{}{\left| P_{ijk}-P_{ijk} \right|}=\sum\nolimits_{i,j}^{}{\left| f\left( \xi _i,\xi _j,\eta _k \right) -f\left( \xi _i,\xi _j,\eta _{\tilde{k}} \right) \right|}   
		\le \sum\nolimits_{i,j}^{}{L_K\Delta _K}=n^2L_K\Delta _K.
	\end{align}
	
	We are now ready to prove the claim 2 of Lemma \ref{lemma2}. We denote the $A_k$ and $P_k$ as the adjacency and probability matrix of the $k$th network respectively. We  bound $ \tilde{d}\left( k,k^{\prime} \right)$ via bounding
	$ \tilde{d}\left( k,\tilde{k} \right) $ for $k^{\prime}\in \mathcal{N}^{k}$. By \eqref{lemma2_1} and \eqref{lemma2_2} , with probability $1-2\exp \left( -n^{C_3/4} \right)$ we have
	\begin{align}\label{lemma2_3}
		& \tilde{d}\left( k,\tilde{k} \right) =\max_{l\ne k,\tilde{k}} \frac{1}{n^2}\left| tr\left( A^kA^l \right) -tr\left( A^{\tilde{k}}A^l \right) \right|  \nonumber \\
		=&\max_{l\ne k,\tilde{k}} \frac{1}{n^2}\left| tr\left( A^kA^l+P^kP^l-P^kP^l-A^{\tilde{k}}A^l+P^{\tilde{k}}P^l-P^kP^l \right) \right|    \nonumber   \\
		\le &\max_{l\ne k,\tilde{k}} \frac{1}{n^2}\left| tr\left( P^kP^l-P^{\tilde{k}}P^l \right) \right|+2\max_{s,t:s\ne t} \frac{1}{n^2}\left| tr\left( A^sA^t-P^sP^t \right) \right|    \nonumber   \\
		=&\max_{l\ne k,\tilde{k}} \frac{1}{n^2}\left| \sum_{i,j}{\left( P_{ijk}P_{ijl}-P_{ij\tilde{k}}P_{ijl} \right)} \right|+2\max_{s,t:s\ne t} \frac{1}{n^2}\left| \sum_{i,j}{\left( A_{ijs}A_{ijt}-P_{ijs}P_{ijt} \right)} \right|    \nonumber   \\
		\le &L_K\Delta _K+\frac{1}{n^{1-C_3}}.
	\end{align}
	Now since the fraction of layers contained in $\left| \left\{ \tilde{k}:\eta _{\tilde{k}}\in S_k\left( \Delta _K \right) \right\} \right|$ is at least $h=C_0\left\{ \left( nK \right) ^{-1}\log n \right\} ^{1/3}$, this upper bounds $ \tilde{d}\left( k,k^{\prime} \right)$ for $k^{\prime}\in \mathcal{N} ^k$, since layer in $\mathcal{N}^{k}$ have the lowest $h$ fraction of values in 
	$\left\{  \tilde{d}\left( k,q \right) \right\} _q$. Setting $\Delta _K$ as $\tilde{C}_0\left\{ \left( nK \right) ^{-1}\log n \right\} ^{1/3}$ and $0<C \le \tilde{C}$, by   \eqref{lemma2_1}, with probability $1-2\exp \left( -n^{C_3/4} \right)$, for all $k$, at least $\tilde{C}_0\left\{ \left( nK \right) ^{-1}\log n \right\} ^{1/3}$  fraction of layer  $\tilde{k}\ne k$ satisfy both $\eta _{\tilde{k}}\in S_k\left( \Delta _K \right) $ and \eqref{lemma2_3}. Then we yields that
	\begin{equation}\label{lemma2_4}
		\tilde{d}\left( k,k^{\prime} \right) =L_K\Delta _K+\frac{1}{n^{1-C_3}}
	\end{equation}
	holds for all $k$ and $k^{\prime}\in \mathcal{N} ^k$ simultaneously with $n$ grows.
	
	We are now ready to complete the proof of the first claim of Lemma 2. By adding and subtracting items, tackling the $\boldsymbol{A}$ by \eqref{lemma2_2}, the $\boldsymbol{B}$ by \eqref{lemma2_1}, the $\boldsymbol{C}$ by \eqref{lemma2_4}, finally we have
	
	
	\begin{align}
		&\sum\nolimits_{i,j}^{}{\left( P_k-P_{k^{\prime}} \right) ^2}\le \left| \sum_{i,j}{\left( {P_k}^2-P_{k^{\prime}}P_k \right)} \right|+\left| \sum_{i,j}{\left( {P_{k^{\prime}}}^2-P_kP_{k^{\prime}} \right)} \right|  \nonumber  \\
		=&\underset{\boldsymbol{A}}{\underbrace{\left| \sum_{i,j}{\left( {P_k}^2-P_kP_{\tilde{k}} \right)} \right|+\left| \sum_{i,j}{\left( P_{k^{\prime}}P_{\tilde{k}}-P_{k^{\prime}}P_k \right)} \right|+\left| \sum_{i,j}{\left( {P_{k^{\prime}}}^2-P_{k^{\prime}}P_{\tilde{k}^{\prime}} \right)} \right|+\left| \sum_{i,j}{\left( P_kP_{\tilde{k}^{\prime}}-P_kP_{k^{\prime}} \right)} \right|}}    \nonumber   \\
		&+\left| \sum_{i,j}{\left( P_kP_{\tilde{k}}-P_{k^{\prime}}P_{\tilde{k}} \right)} \right|+\left| \sum_{i,j}{\left( P_{k^{\prime}}P_{\tilde{k}^{\prime}}-P_{k^{\prime}}P_{\tilde{k}} \right)} \right|    \nonumber   \\
		\le &4n^2L_K\Delta _K+\left| \sum_{i,j}{\left( P_kP_{\tilde{k}}-P_{k^{\prime}}P_{\tilde{k}} \right)} \right|+\left| \sum_{i,j}{\left( P_{k^{\prime}}P_{\tilde{k}^{\prime}}-P_kP_{\tilde{k}^{\prime}} \right)} \right|    \nonumber\\
		\le &\underset{\boldsymbol{B}}{\underbrace{\left| \sum_{i,j}{\left( P_kP_{\tilde{k}}-A_kA_{\tilde{k}} \right)} \right|+\left| \sum_{i,j}{\left( P_{k^{\prime}}P_{\tilde{k}}-A_{k^{\prime}}A_{\tilde{k}} \right)} \right|+\left| \sum_{i,j}{\left( P_{k^{\prime}}P_{\tilde{k}^{\prime}}-A_{k^{\prime}}A_{\tilde{k}^{\prime}} \right)} \right|+\left| \sum_{i,j}{\left( P_kP_{\tilde{k}^{\prime}}-A_kA_{\tilde{k}^{\prime}} \right)} \right|}}   \nonumber   \\
		&+\underset{\mathcal{C}}{\underbrace{\left| \sum_{i,j}{\left( A_AkA_{\tilde{k}}-A_{k^{\prime}}A_{\tilde{k}} \right)} \right|+\left| \sum_{i,j}{\left( A_{k^{\prime}}A_{\tilde{k}^{\prime}}-A_kA_{\tilde{k}^{\prime}} \right)} \right|}}+4n^2L_K\Delta _K     \nonumber   \\
		\le &4n^{1+C_3}+2\underset{l\ne k,k^{\prime}}{\max}\left| tr\left( A^kA^l \right) -tr\left( A^{k^{\prime}}A^l \right) \right|+4n^2L_K\Delta _K     \nonumber   \\
		=&4n^{1+C_3}+2n^2 \tilde{d}\left( k,k^{\prime} \right) +4n^2L_K\Delta _K \nonumber   \\
		=&8n^{1+C_3}+6n^2L_K\Delta _K
	\end{align}
	This completes the proof of Lemma \ref{lemma2}.
\end{proof}


We continue focusing on the case where $K\ge C^{-3/2}\sqrt{\frac{n}{\log n}}$, and the following lemma shows the estimation error bound in this case.
\begin{lemma}\label{lemma3}
	Assume $L_n$ and $L_K$ in $L$ are global constants, and $\delta =\delta (n,K)$ depends on $n$, $K$ satisfying ~\\ $\lim_{n,K\rightarrow \infty} \delta /\left\{ \left( nK \right) ^{-1}\log n \right\} ^{1/3}\rightarrow \infty $ and $K\ge C^{-3/2}\sqrt{\frac{n}{\log n}}$, then the estimator $\tilde{P}$ defined in \eqref{equation::MNSestimator}, with neighborhood $\mathcal{N}^{k}$, $\mathcal{N} _{i}^{k^{\prime}}$defined in \eqref{equation::layer_neighbor} and \eqref{equation::MNSnode_neighbor} and $h=C\left\{ \left( nK \right) ^{-1}\log n \right\} ^{1/3}$ for any global constant $C > 0$ satisfies 
	\begin{equation*}
		\max_{f\in \mathcal{F} _{\delta ;L}} \mathrm{pr}\left\{ \frac{1}{n^2}\left\| \tilde{P}^k-P^k \right\| _{F}^{2}\ge \tilde{C_1}\left( \frac{\log n}{nK} \right) ^{1/3}+\frac{\tilde{C_3}}{n^{1-\alpha}} \right\} \le n^{-\gamma}
	\end{equation*}
	for $\gamma>0$, where $\tilde{C_1}$ and $\tilde{C_3}$ are some positive global constants, and $0<\alpha<1$ is arbitrary.
\end{lemma}
\begin{proof}[Proof of Lemma \ref{lemma3}]
we first work with the $J_1\left( i,j,k \right)$  to split it into square terms and cross terms, note that 
\begin{equation}
	\begin{aligned}
		&\left\{ \sum_{\substack{k^{\prime}\in \mathcal{N} ^k\\ i^{\prime}\in \mathcal{N} _{i}^{k^{\prime}}}} \left( A_{i^{\prime}jk^{\prime}}-P_{i^{\prime}jk^{\prime}} \right) \right\} ^2 = \sum_{\substack{k^{\prime}\in \mathcal{N} ^k\\ i^{\prime}\in \mathcal{N} _{i}^{k^{\prime}}}} \left( A_{i^{\prime}jk^{\prime}}-P_{i^{\prime}jk^{\prime}} \right) ^2 \\
		&+ \sum_{\substack{k^{\prime}\in \mathcal{N} ^k\\ i^{\prime}\in \mathcal{N} _{i}^{k^{\prime}}}} \left( A_{i^{\prime}jk^{\prime}}-P_{i^{\prime}jk^{\prime}} \right) \sum_{\substack{k^{\prime\prime}\in \mathcal{N} ^k, i^{\prime\prime}\in \mathcal{N} _{i}^{k^{\prime}}\\ k^{\prime\prime}\ne k^{\prime} \text{ or } i^{\prime\prime}\ne i^{\prime}}} \left( A_{i^{\prime}jk^{\prime\prime}}-P_{i^{\prime\prime}jk^{\prime\prime}} \right).
	\end{aligned}
\end{equation}
then it can be decomposed  as follows:
\begin{equation}\label{1termdecomposed}
	\begin{aligned}
		\sum_{i,j} J_1(i,j,k) &= \sum_{i,j} \left\{ \sum_{\substack{k^{\prime}\in \mathcal{N} ^k\\ i^{\prime}\in \mathcal{N} _{i}^{k^{\prime}}}} \left( A_{i^{\prime}jk^{\prime}}-P_{i^{\prime}jk^{\prime}} \right) \right\} ^2
		= \sum_i \left\{ \sum_j \sum_{\substack{k^{\prime}\in \mathcal{N}^k\\ i^{\prime}\in \mathcal{N}_{i}^{k}}} \left( A_{i^{\prime}jk^{\prime}}-P_{i^{\prime}jk^{\prime}} \right) ^2 \right. \\
		&+ \left. \sum_j \sum_{\substack{k^{\prime}\in \mathcal{N}^k\\ i^{\prime}\in \mathcal{N}_{i}^{k}}} \left( A_{i^{\prime}jk^{\prime}}-P_{i^{\prime}jk^{\prime}} \right) \sum_{\substack{k^{\prime\prime}\in \mathcal{N}^k, i^{\prime\prime}\in \mathcal{N}_{i}^{k}\\ k^{\prime\prime}\ne k^{\prime} \text{ or } i^{\prime\prime}\ne i^{\prime}}} \left( A_{i^{\prime\prime}jk^{\prime\prime}}-P_{i^{\prime\prime}jk^{\prime\prime}} \right) \right\}.
	\end{aligned}
\end{equation}
The first term in \eqref{1termdecomposed} satisfies
\begin{equation}\label{secondterm1}
	\sum_{\substack{k^{\prime}\in \mathcal{N}^k\\ i^{\prime}\in \mathcal{N}_{i}^{k^{\prime}}}}
	\sum_j \left( A_{i^{\prime}jk^{\prime}} - P_{i^{\prime}jk^{\prime}} \right)^2 
	= \sum_{\substack{k^{\prime}\in \mathcal{N}^k\\ i^{\prime}\in \mathcal{N}_{i}^{k^{\prime}}}}
	\left\| A_{i^{\prime}\cdot k^{\prime}} - P_{i^{\prime}\cdot k^{\prime}} \right\| _{2}^{2} 
	\le \left| \mathcal{N}_{i}^{k^{\prime}} \right| \cdot \left| \mathcal{N}^k \right|
\end{equation}
where the inequality is due to $\left| A_{i^{\prime}jk^{\prime}}-P_{i^{\prime}jk^{\prime}} \right|\le \left| e_{i^{\prime}jk^{\prime}} \right|=1$ for all $j$.

To bound the second term in \eqref{1termdecomposed}, we could consider the following expression:
\begin{equation}\label{secondterm2}
	\sum_{\substack{i^{\prime}\in \mathcal{N}_{i}^{k^{\prime}}\\ i^{\prime\prime}\in \mathcal{N}_{i}^{k^{\prime}}}}
	\left| 
	\sum_{\substack{k^{\prime}\in \mathcal{N}^{k}, k^{\prime\prime}\in \mathcal{N}^{k}\\ k^{\prime\prime}\ne k^{\prime}}}
	\sum_j \left( A_{i^{\prime}jk^{\prime}} - P_{i^{\prime}jk^{\prime}} \right) \left( A_{i^{\prime\prime}jk^{\prime\prime}} - P_{i^{\prime\prime}jk^{\prime\prime}} \right)
	\right|
\end{equation}
for any $k^{\prime\prime}\ne k^{\prime},j $ and $ 0<\epsilon \le 1$, by Bernstein's inequality we have
\begin{align*}
	&P\left\{ \frac{1}{n\left| \mathcal{N} ^k \right|\left( \left| \mathcal{N} ^k \right|-1 \right)}\left| \sum_{\begin{array}{c}
			k^{\prime}\in \mathcal{N} ^k,k^{\prime\prime}\in \mathcal{N} ^k\\
			k^{\prime\prime}\ne k^{\prime}\\
	\end{array}}{\sum_j{\left( A_{i^{\prime}jk^{\prime}}-P_{i^{\prime}jk^{\prime}} \right) \left( A_{i^{\prime\prime}jk^{\prime\prime}}-P_{i^{\prime\prime}jk^{\prime\prime}} \right)}} \right|\ge \epsilon \right\}\\
	\le &2\exp \left\{ \frac{-n\left| \mathcal{N} ^k \right|\left( \left| \mathcal{N} ^k \right|-1 \right) \epsilon ^2}{2\left( 1+1/\epsilon \right)} \right\} \approx 2\exp \left\{ \frac{-n\left| \mathcal{N} ^k \right|^2\epsilon ^2}{4} \right\}.
\end{align*}
Taking a union bound over all $ i^{\prime}$ and $i^{\prime\prime} $, let $C_1$ be an arbitrary global constants and let $n$ be large enough. Then by setting $\epsilon=\left\{ \left( C_1+8 \right) \log n/n|\mathcal{N}^k|^2 \right\} ^{1/2}$, we have
\begin{align}\label{secondterm3}
	&P\left\{ \underset{i^{\prime},i^{''}}{\max}\,\,\frac{1}{n\left| \mathcal{N} ^k \right|\left( \left| \mathcal{N} ^k \right|-1 \right)}\left| \sum_{\begin{array}{c}
			k^{\prime}\in \mathcal{N} ^k,k^{\prime\prime}\in \mathcal{N} ^k\\
			k^{\prime\prime}\ne k^{\prime}\\
	\end{array}}{\sum_j{\left( A_{i^{\prime}jk^{\prime}}-P_{i^{\prime}jk^{\prime}} \right) \left( A_{i^{\prime\prime}jk^{\prime\prime}}-P_{i^{\prime\prime}jk^{\prime\prime}} \right)}} \right|\ge \epsilon \right\}    \nonumber \\
	&\le 2n^2\exp \left\{ \frac{-n\left| \mathcal{N} ^k \right|^2\epsilon ^2}{4} \right\} =2n^{-C_1/4}\rightarrow 0.
\end{align}
Then we can bound  \eqref{secondterm2} as follows:
\begin{equation}\label{secondterm4}
	\begin{aligned}
		&\sum_{\substack{i^{\prime}\in \mathcal{N}_{i}^{k^{\prime}}\\ i^{\prime\prime}\in \mathcal{N}_{i}^{k^{\prime}}}}
		\left| 
		\sum_{\substack{k^{\prime}\in \mathcal{N}^{k}, k^{\prime\prime}\in \mathcal{N}^{k}\\ k^{\prime\prime}\ne k^{\prime}}}
		\sum_j \left( A_{i^{\prime}jk^{\prime}} - P_{i^{\prime}jk^{\prime}} \right) \left( A_{i^{\prime\prime}jk^{\prime\prime}} - P_{i^{\prime\prime}jk^{\prime\prime}} \right)
		\right| \\
		\le &n\left| \mathcal{N}_{i}^{k^{\prime}} \right|^2 \left| \mathcal{N}^{k} \right| \left( \left| \mathcal{N}^{k} \right| - 1 \right) 
		\left( \frac{(C_1 + 8)\log n}{n\left| \mathcal{N}^{k} \right|^2} \right)^{1/2}
		\approx n\left| \mathcal{N}_{i}^{k^{\prime}} \right|^2 \left| \mathcal{N}^{k} \right|^2 
		\left( \frac{(C_1 + 8)\log n}{n\left| \mathcal{N}^{k} \right|^2} \right)^{1/2}
	\end{aligned}
\end{equation}
Plugging \eqref{secondterm1}, \eqref{secondterm3} and \eqref{secondterm4} into \eqref{decomposition}, with probability $1-2n^{C_1/4}$, for all $(i,j)$ simultaneously, we have
\begin{align}\label{J1}
	\frac{1}{n^2M^2}&\sum\nolimits_{i,j}^{}{J_1\left( i,j,k \right)}=\frac{1}{n^2M^2}\sum\nolimits_i^{}{\left\{ \left| \mathcal{N} ^k \right|\left| \mathcal{N} _{i}^{k^{\prime}} \right|+n\left| \mathcal{N} _{i}^{k^{\prime}} \right|^2\left| \mathcal{N} ^k \right|^2\left( \frac{(C_1+8)\log n}{n\left| \mathcal{N} ^k \right|^2} \right) ^{1/2} \right\}}  \nonumber\\
	&\approx \frac{n}{n^2\left| \mathcal{N} ^k \right|^2\left| \mathcal{N} _{i}^{k^{\prime}} \right|^2}\left\{ \left| \mathcal{N} ^k \right|\left| \mathcal{N} _{i}^{k^{\prime}} \right|+n\left| \mathcal{N} _{i}^{k^{\prime}} \right|^2\left| \mathcal{N} ^k \right|^2\left( \frac{(C_1+8)\log n}{n\left| \mathcal{N} ^k \right|^2} \right) ^{1/2} \right\} \nonumber\\
	&=\frac{1}{n\left| \mathcal{N} ^k \right|\left| \mathcal{N} _{i}^{k^{\prime}} \right|}+\left( \frac{C_1+8}{K} \right) ^{1/2}\left( \frac{\log n}{nK} \right) ^{1/6}
\end{align}
We temporarily retain this form for the first term. In the subsequent proofs, we will demonstrate that this term does not influence the final results. 

For the  $J_2\left( i,j,k \right)$, by Lemma \ref{lemma1} given before, we can bound it with probability  $1-2n^{-\tilde{C_2}/4}$ for all $(i,j)$ simultaneously, we have
\begin{align}\label{J2}
	&\frac{1}{n^2M^2}\sum\nolimits_{i,j}^{}{J_2\left( i,j,k \right)}=\frac{1}{n^2M^2}\sum\nolimits_{i,j}^{}{\left( \sum\nolimits_{k^{\prime},i^{\prime}}^{}{\left( P_{i^{\prime}jk^{\prime}}-P_{ijk^{\prime}} \right)} \right) ^2}   \nonumber  \\
	\le &\frac{1}{\left( n\left| \mathcal{N} _{i}^{k^{\prime}} \right|\left| \mathcal{N} ^k \right| \right) ^2}\sum\nolimits_i^{}{\left| \mathcal{N} _{i}^{k^{\prime}} \right|\left| \mathcal{N} ^k \right|}\sum\nolimits_{i^{\prime}\in \mathcal{N} _{i}^{k^{\prime}}}^{}{\sum\nolimits_{k^{\prime}\in \mathcal{N} ^k,j}^{}{\left( P_{i^{\prime}jk^{\prime}}-P_{ijk^{\prime}} \right) ^2}}    \nonumber   \\
	\le &\frac{1}{\left( n\left| \mathcal{N} _{i}^{k^{\prime}} \right|\left| \mathcal{N} ^k \right| \right) ^2}\sum\nolimits_i^{}{\left| \mathcal{N} _{i}^{k^{\prime}} \right|\left| \mathcal{N} ^k \right|}\sum\nolimits_{i^{\prime}\in \mathcal{N} _{i}^{k^{\prime}}}^{}{n\left| \mathcal{N} ^k \right|\left[ 6L_n\Delta _n+8\left\{ \frac{\left( \tilde{C}_2+\tilde{C_2}+8 \right) \log n}{nK} \right\} ^{1/3} \right]}  \nonumber  \\
	=&\left\{ 6\tilde{C}_0L_n+8\left( \tilde{C_2}+8 \right) ^{1/3} \right\} \left( \frac{\log n}{nK} \right) ^{1/3}+\frac{32}{n},
\end{align}
where  $J_2\left( i,j,k \right) =\left( \sum\nolimits_{k^{\prime},i^{\prime}}^{}{\left( P_{i^{\prime}jk^{\prime}}-P_{ijk^{\prime}} \right)} \right) ^2\le \left| \mathcal{N} _{i}^{k^{\prime}} \right|\left| \mathcal{N} ^k \right|\sum\nolimits_{k^{\prime},i^{\prime}}^{}{\left( P_{i^{\prime}jk^{\prime}}-P_{ijk^{\prime}} \right) ^2}$.

For the last term $J_3\left( i,j,k \right)$, we can decomposed it by Cauchy-Schwarz inequality:
\begin{align*}
	&J_3\left( i,j,k \right) =\left( \sum\nolimits_{k^{\prime},i^{\prime}}^{}{\left( P_{ijk^{\prime}}-P_{ijk} \right)} \right) ^2=\left( \sum\nolimits_{k^{\prime}\in \mathcal{N} ^k}^{}{\left| \mathcal{N} _{i}^{k^{\prime}} \right|\left( P_{ijk^{\prime}}-P_{ijk} \right)} \right) ^2\\
	\le& \left[ \sum\nolimits_{k^{\prime}\in \mathcal{N} ^k}^{}{\left| \mathcal{N} _{i}^{k^{\prime}} \right|^2} \right] \left[ \sum\nolimits_{k^{\prime}\in \mathcal{N} ^k}^{}{\left( P_{ijk^{\prime}}-P_{ijk} \right) ^2} \right]
	=\left| \mathcal{N} ^k \right|\left| \mathcal{N} _{i}^{k^{\prime}} \right|^2\sum\nolimits_{k^{\prime}\in \mathcal{N} ^k}^{}{\left( P_{ijk^{\prime}}-P_{ijk} \right) ^2}
\end{align*}

Then we can bound $J_3\left( i,j,k \right)$ by Lemma \ref{lemma2}, with probability $1-2\exp \left( -n^{C_3/4} \right)$ for all $(i,j)$ simultaneously, we have
\begin{align}\label{J3}
	\frac{1}{n^2M^2}\sum\nolimits_{i,j}^{}{J_3\left( i,j,k \right)}&\le \frac{1}{n^2M^2}\sum\nolimits_{i,j}^{}{\left| \mathcal{N} ^k \right|\left| \mathcal{N} _{i}^{k^{\prime}} \right|^2\sum\nolimits_{k^{\prime}}^{}{\left( P_{ijk^{\prime}}-P_{ijk} \right) ^2}}   \nonumber\\
	&=\frac{\left| \mathcal{N} ^k \right|\left| \mathcal{N} _{i}^{k^{\prime}} \right|^2}{n^2\left| \mathcal{N} ^k \right|^2\left| \mathcal{N} _{i}^{k^{\prime}} \right|^2}\sum\nolimits_{k^{\prime}}^{}{\sum\nolimits_{i,j}^{}{\left( P_{ijk^{\prime}}-P_{ijk} \right) ^2}}\nonumber\\
	&=\frac{1}{n^2}\left[ 8n^{1+C_3}+6n^2\tilde{C}_0L_K\left( \frac{\log n}{nK} \right) ^{1/3} \right]\nonumber\\
	&=\frac{8}{n^{1-C_3}}+6\tilde{C}_0L_K\left( \frac{\log n}{nK} \right) ^{1/3}.
\end{align}

Combining \eqref{J1}, \eqref{J2} and \eqref{J3} 
and when $K\gtrsim  \left( n/\log n \right) ^{1/2}$, then we have
\begin{equation}\label{Ksmall}
	\begin{split}
		\frac{1}{n^2}\sum\nolimits_{i,j}^{}{\left( \tilde{P}_{ijk}-P_{ijk} \right) ^2}&\le \frac{3}{n^2M^2}\sum_{i,j}{\left\{ J_1\left( i,j,k \right) +J_2\left( i,j,k \right) +J_3\left( i,j,k \right) \right\}}\\
		&\le \tilde{C_1}\left( \frac{\log n}{nK} \right) ^{1/3}+\frac{\tilde{C_3}}{n^{1-\alpha}}.
	\end{split}	
\end{equation}
This completes the proof of Lemma \ref{lemma3}.
\end{proof}

\begin{center}
	\textbf{Proof of Theorem \ref{theorem1}}
\end{center}
\begin{proof}
	According to the condition $K\lesssim  n^2\log n$ in Theorem \ref{theorem1}, The first term in \refeq{Ksmall} is the dominating order. Thus we only need to derive the upper bound when $K< C^{-3/2}\sqrt{\frac{n}{\log n}}$.  It is observed that this scenario implies $h_1=1/K$ and $h_2=C^{\frac{3}{2}}\sqrt{\frac{\log n}{n}}$. Under such conditions, the proposed method employs only a single-layer network, thereby degenerating to the approach of \cite{zhang2017estimating}. According to Theorem 1 in \cite{zhang2017estimating} with $h =C^{-3/2}\sqrt{\frac{n}{\log n}}$, there exist constants $C_4>0$ and $C_5>0$ such that, with probability at least $1 - n^{-C_5}$, we have
	\begin{equation}\label{K1}
		\frac{1}{n^2}\sum\nolimits_{i,j}^{}{\left( \tilde{P}_{ijk}-P_{ijk} \right) ^2}\le C_4\left(\frac{\log n}{n}\right)^{\frac{1}{2}}.
	\end{equation}
Combining \eqref{Ksmall} and \eqref{K1} completes the proof of Theorem \ref{theorem1}.
\end{proof}

\begin{center}
	\textbf{Proof of Theorem \ref{theorem2}}
\end{center}
\begin{proof}
Notice $K> n^2\log n$ indicates the second term in \refeq{Ksmall} becomes the dominant component, thereby establishing Theorem \ref{theorem2}.
\end{proof}

\section{Discussion}
We propose a model-free, scalable method for estimating probability matrices in multi-layer networks. Our approach leverages adaptive neighborhood selection to fully utilize similarity information across network layers, thereby enhancing the estimation accuracy of the probability matrix for each given layer. The proposed MNS method is computationally efficient and requires little tuning. 

While the proposed method is proven to be consistent, achieving the convergence rate $O\left( \left\{ \log n/\left( nK \right) \right\} ^{1/3} \right)$ requires a large network size, which may be challenging in practical applications. This study relies on the independence across layers in multi-layer networks. Future work could relax this assumption to better accommodate the complexities of real-world networks. Additionally, with advancements in tensor modeling, our approach could be extended to higher-order tensors, integrating tensor-based methods for broader applicability. In theory, we have not demonstrated the minimax optimality of our method. Establishing a general lower bound is an important direction for future research.

\section*{Acknowledgment}
Li' work is supported by Zhejiang Provincial Natural Science Foundation of China (LQ23A010008), the Summit Advancement Disciplines of Zhejiang Province (Zhejiang Gongshang University - Statistics), and Collaborative Innovation Center of Statistical Data  Engineering Technology \& Application.






\bibliographystyle{elsarticle-num-names}
\begin{spacing}{0.8}
\bibliography{paper-ref}

@Article{Aldous1981representations,
  title={Representations for partially exchangeable arrays of random variables},
  author={Aldous, David J},
  journal={Journal of Multivariate Analysis},
  volume={11},
  number={4},
  pages={581--598},
  year={1981},
  publisher={Elsevier}
}

@Article{zhang2017estimating,
  title={Estimating network edge probabilities by neighbourhood smoothing},
  author={Zhang, Yuan and Levina, Elizaveta and Zhu, Ji},
  journal={Biometrika},
  volume={104},
  number={4},
  pages={771--783},
  year={2017},
  publisher={Oxford University Press}
}

@article{minimaxGao2015,
  author = {Gao, C. and Lu, Y. and Zhou, H. H.},
  title = {Rate-optimal graphon estimation},
  journal = {The Annals of Statistics},
  year = {2015},
  volume = {43},
  issue = {6},
  pages = {2624--2652},
}

@Article{dataset2010,
  author={Manlio De Domenico and Vincenzo Nicosia and Alexandre Arenas and Vito Latora},
  title={{Structural reducibility of multilayer networks}},
  journal={Nature Communications},
  year=2015,
  volume={6},
  number={1},
  pages={1-9},
  month={November},

}

@unpublished{hoover1979relations,
  author       = {Hoover, David N.},
  title        = {Relations on Probability Spaces and Arrays of Random Variables},
  note         = {Preprint, Institute for Advanced Study, Princeton},
  year         = {1979}
}

@article{YubaiYuan,
author = {Yubai Yuan and Annie Qu},
title = {{Community detection with dependent connectivity}},
volume = {49},
journal = {The Annals of Statistics},
number = {4},
publisher = {Institute of Mathematical Statistics},
pages = {2378 -- 2428},
year = {2021},
}

@article{Dynamicgraphon,
author = {Marianna Pensky},
title = {{Dynamic network models and graphon estimation}},
volume = {47},
journal = {The Annals of Statistics},
number = {4},
publisher = {Institute of Mathematical Statistics},
pages = {2378 -- 2403},
year = {2019},

}

@article{changepointdetection,
author = {Daren Wang and Yi Yu and Alessandro Rinaldo},
title = {{Optimal change point detection and localization in sparse dynamic networks}},
volume = {49},
journal = {The Annals of Statistics},
number = {1},
publisher = {Institute of Mathematical Statistics},
pages = {203 -- 232},
year = {2021},

}

@article{hypergraphons,
  author  = {Krishnakumar Balasubramanian},
  title   = {Nonparametric Modeling of Higher-Order Interactions via Hypergraphons},
  journal = {Journal of Machine Learning Research},
  year    = {2021},
  volume  = {22},
  number  = {146},
  pages   = {1--35},

}

@article{Zhao2019ChangepointDI,
  title={Change-point detection in dynamic networks via graphon estimation},
  author={Zifeng Zhao and Li Chen and Lizhen Lin},
  journal={arXiv: 1908.01823},
  year={2019}
}

@article{xu2023covariate,
  title={Covariate-assisted community detection in multi-layer networks},
  author={Xu, Shirong and Zhen, Yaoming and Wang, Junhui},
  journal={Journal of Business \& Economic Statistics},
  volume={41},
  number={3},
  pages={915--926},
  year={2023},
  publisher={Taylor \& Francis}
}

@article{jing2021community,
  title={Community detection on mixture multilayer networks via regularized tensor decomposition},
  author={Jing, Bing-Yi and Li, Ting and Lyu, Zhongyuan and Xia, Dong},
  journal={The Annals of Statistics},
  volume={49},
  number={6},
  pages={3181--3205},
  year={2021},
  publisher={Institute of Mathematical Statistics}
}

@ARTICLE{9634845,
  author={Pham, Phu and Nguyen, Loan T. T. and Nguyen, Ngoc Thanh and Pedrycz, Witold and Yun, Unil and Vo, Bay},
  journal={IEEE Transactions on Systems, Man, and Cybernetics: Systems}, 
  title={ComGCN: Community-Driven Graph Convolutional Network for Link Prediction in Dynamic Networks}, 
  year={2022},
  volume={52},
  number={9},
  pages={5481-5493},
}

@article{Paul2015ConsistentCD,
  title={Consistent community detection in multi-relational data through restricted multi-layer stochastic blockmodel},
  author={Subhadeep Paul and Yuguo Chen},
  journal={Electronic Journal of Statistics},
  year={2015},
  volume={10},
  pages={3807-3870},

}

@article{Wolfe2013NonparametricGE,
  title={Nonparametric graphon estimation},
  author={Patrick J. Wolfe and Sofia Charlotta Olhede},
  journal={arXiv: 1309.5936},
  year={2013},

}

@article{Chan2014ACH,
author = {Chan, Stanley and Airoldi, Edoardo},
year = {2014},
month = {02},
pages = {208--216},
title = {A Consistent Histogram Estimator for Exchangeable Graph Models},
volume = {32},
journal = {Journal of
Machine Learning Research Workshop and Conference Proceedings}
}

@article{doi:10.1073/pnas.1400374111,
    author = {Sofia C. Olhede  and Patrick J. Wolfe },
    title = {Network histograms and universality of blockmodel approximation},
    journal = {Proceedings of the National Academy of Sciences},
    volume = {111},
    number = {41},
    pages = {14722-14727},
    year = {2014},

    }

@article{10.1214/16-AOS1497,
    author = {David Choi},
    title = {{Co-clustering of nonsmooth graphons}},
    volume = {45},
    journal = {The Annals of Statistics},
    number = {4},
    publisher = {Institute of Mathematical Statistics},
    pages = {1488 -- 1515},
    year = {2017},

    }

@article{10.1093/biomet/asab057,
    author = {Chandna, S and Olhede, S C and Wolfe, P J},
    title = {Local linear graphon estimation using covariates},
    journal = {Biometrika},
    volume = {109},
    number = {3},
    pages = {721-734},
    year = {2021},
    month = {11},
    issn = {1464-3510},
}

@article{macdonald2022latent,
  title={Latent space models for multiplex networks with shared structure},
  author={MacDonald, Peter W and Levina, Elizaveta and Zhu, Ji},
  journal={Biometrika},
  volume={109},
  number={3},
  pages={683--706},
  year={2022},
  publisher={Oxford University Press}
}

@article{Chandna2020NonparametricRF,
  title={Nonparametric regression for multiple heterogeneous networks},
  author={Swati Chandna and Pierre-Andr{\'e} G. Maugis},
  journal={arXiv: Methodology},
  year={2020},
}
\end{spacing}

\end{document}